\documentclass{sig-alternate-2013}

\newfont{\mycrnotice}{ptmr8t at 7pt}
\newfont{\myconfname}{ptmri8t at 7pt}
\let\crnotice\mycrnotice%
\let\confname\myconfname%

\permission{Permission to make digital or hard copies of all or part of this work for personal or classroom use is granted without fee provided that copies are not made or distributed for profit or commercial advantage and that copies bear this notice and the full citation on the first page. Copyrights for components of this work owned by others than ACM must be honored. Abstracting with credit is permitted. To copy otherwise, or republish, to post on servers or to redistribute to lists, requires prior specific permission and/or a fee. Request permissions from permissions@acm.org.}
\conferenceinfo{PODS'15,}{May 31--June 4, 2015, Melbourne, Victoria, Australia.\\
{\mycrnotice{Copyright is held by the owner/author(s). Publication rights licensed to ACM.}}}
\copyrightetc{ACM \the\acmcopyr}
\crdata{978-1-4503-2757-2/15/05\ ...\$15.00.\\
http://dx.doi.org/10.1145/2745754.2745765}
\usepackage{amsmath}

\usepackage{amsthm}
\usepackage{amsfonts}
\usepackage{hyperref}
\usepackage{times} 
\usepackage[vlined,boxed,ruled,linesnumbered]{algorithm2e}
\usepackage{xspace}
\usepackage{graphicx}
\usepackage{multirow}
\usepackage{float}
\usepackage{enumitem}
\usepackage{pbox}
\usepackage{tikz}
\usepackage{tikz-qtree}
\usepackage{thmtools}
\usepackage{thm-restate}
\usepackage{cleveref}
 \usepackage{url}
 \usepackage{balance}

\usepackage{enumitem}
\setlist[itemize]{leftmargin=1.5em, itemsep=1pt, topsep=1pt}
\setlist[enumerate]{leftmargin=1.5em, itemsep=1pt, topsep=1pt}

\newtheorem{theorem}{Theorem}[section]
\newtheorem{proposition}[theorem]{Proposition}
\newtheorem{claim}[theorem]{Claim}
\newtheorem{definition}[theorem]{Definition}
\newtheorem{corollary}[theorem]{Corollary}
\newtheorem{problem}[theorem]{Problem}

\newtheorem{lemma}[theorem]{Lemma}
\theoremstyle{definition}
\newtheorem{example}[theorem]{Example}

\newcommand{\sch}{\mathbf{S}} 
\newcommand{\dsch}{\mathbf{D}} 
\renewcommand{\phi}{\varphi}
\newcommand{\adom}[1]{\mathrm{adom}(#1)}
\newcommand{\ex}[2]{[\![#1]\!]^{#2}}	
 
\newcommand{\exs}[2]{ext(#1,#2)} 
\newcommand{\concepts}{{\cal C}}
\newcommand{\sontoshort}{({\cal C},\sqsubseteq,ext)}

\newcommand{\sontology}[3]{(#1,#2,#3)}
\newcommand{\tup}{(a_1,\ldots,a_m)}

\newcommand{\das}{\sch}
\newcommand{\vsch}{\mathbf{V}}
\newcommand{\tuple}[1]{\overline{#1}}
\newcommand{\inst}[1]{\mathrm{Inst}(#1)}
\newcommand{\const}{\mathbf{Const}}

\newcommand{\nom}[1]{\{#1\}}

\newcommand{\obda}{\mathcal{B}}
\newcommand{\bspec}{(\mathcal{T},\sch,\mathcal{M})}
\newcommand{\wninstance}{(\das,I,q,Ans,\tuple{a})}
\newcommand{\naivealgo}{\textsc{Exhaustive Search Algorithm}\xspace}
\newcommand{\greedyalgo}{\textsc{Incremental Search Algorithm}\xspace}
\newcommand{\greedyalgosel}{\textsc{Incremental Search Algorithm with Selections}\xspace}

\newcommand{\dllite}{\textit{DL-Lite}$_\mathcal{R}$ }
\newcommand{\dlliteshort}{\textit{DL-Lite}$_\mathcal{R}$}
\newcommand {\CQ}{\mathsf{CQ}}
\newcommand{\UCQ}{\textrm{UCQ}}
\newcommand{\tupp}[1]{\langle #1\rangle}
\newcommand{\ontocal}{\mathcal{O}}
\newcommand{\calK}{\mathcal{K}}
\newcommand{\LWmin}{L_\das^{\mathrm{min}}}
\newcommand{\dwhynotp}{\textsc{Existence-of-explanation}\xspace}
\newcommand{\whynotp}{\textsc{Compute-one-MGE}\xspace}
\newcommand{\mgep}{\textsc{Check-MGE}\xspace}
\newcommand{\dwhynotproblem}[1]{\textsc{Existence-of-explanation w.r.t.~#1}\xspace}
\newcommand{\whynotproblem}[1]{\textsc{Compute-one-MGE w.r.t.~#1}\xspace}
\newcommand{\mgeproblem}[1]{\textsc{Check-MGE w.r.t.~#1}\xspace}

\newcounter{tbsnr}
\newenvironment{tbs}
{\addtocounter{tbsnr}{1}\par\bigskip \noindent\fbox{\thetbsnr}
\hspace{2mm}\begin{minipage}{.9\linewidth}\tt \small}
{\end{minipage}\hspace*{\fill}\bigskip}

\newcommand{\ACz}{\textsc{AC}$^0$\xspace}

\newcommand{\PTIME}{\textsc{PTime}\xspace}

\newcommand{\EXPTIME}{\textsc{ExpTime}\xspace}

\newcommand{\coNEXPTIME}{\textsc{coNExpTime}\xspace}
\newcommand{\NP}{\textsc{NP}\xspace}

\newcommand{\piptwo}{\Pi^P_2}

\newcommand{\lubsigma}[1]{\textsf{lub}^\sigma_I(#1)}

\declaretheoremstyle[bodyfont=\normalfont\itshape]{thstyle}

\begin{document}

\title{High-Level Why-Not Explanations using Ontologies}

\numberofauthors{4}
\author{
\hspace*{-1em}
Balder ten Cate\\
	\hspace*{-1em}\affaddr{LogicBlox and UCSC}\\
	\hspace*{-1em}\email{balder.tencate@gmail.com}
\and \hspace*{-1em}
Cristina Civili\\
	\hspace*{-1em}\affaddr{Sapienza Univ. of Rome}\\
	\hspace*{-1em}\email{civili@dis.uniroma1.it}
\and \hspace*{-1em}
 Evgeny Sherkhonov\\
 	\hspace*{-1em}\affaddr{Univ. of Amsterdam}\\
	\hspace*{-1em}\email{e.sherkhonov@uva.nl}
\and \hspace*{-1em}
Wang-Chiew Tan\\
	\hspace*{-0.5em}\affaddr{UCSC}\\
	\hspace*{-0.5em}\email{tan@cs.ucsc.edu}
}
\date{}

\maketitle

\begin{abstract}

We propose a novel foundational framework for 
\emph{why-not explanations}, that is, 
explanations for why a tuple is missing from a query result.
Our why-not explanations leverage concepts from an 
ontology to provide high-level and meaningful reasons for why a 
tuple is missing from the result of a query. 

A key algorithmic problem in our framework is that of
\emph{computing a most-general explanation} for a why-not question,
relative to an ontology, which can either be provided by the user, or
it may be automatically derived from the data and/or schema.
We study the complexity of this problem and associated problems,
and present concrete algorithms for computing why-not explanations.  
In the case where an external ontology is provided,
we first show that the problem of deciding the existence of an explanation 
to a why-not question is
\NP-complete in general. However, the problem is solvable in polynomial time for queries
of bounded arity, provided that the ontology is specified in a 
suitable language, such as a member of the DL-Lite family of description logics, which allows for 
efficient concept subsumption checking.
Furthermore, we show that a most-general explanation can be computed
in polynomial time in this case.
In addition, we propose a method for deriving a suitable (virtual) ontology 
from a database and/or a schema, 
and we present an algorithm for computing 
a most-general explanation to a why-not question, relative 
to such ontologies.
This algorithm runs in polynomial-time 
in the case when concepts are defined in a selection-free language,
or if the underlying schema is fixed. 
Finally, we also study the problem 
of computing \emph{short} most-general
explanations, and we briefly discuss alternative definitions of
what it means to be an explanation, and to be most general.

\end{abstract}

\vspace{-1mm}

\vspace{-1mm}

\vspace{-1mm}

\section{Introduction and Results}
\label{sec:Introduction}

An increasing number of databases are derived, extracted,
or curated from disparate data sources. 
Consequently, it becomes more and more important to provide 
data consumers with 
mechanisms that will allow them to gain an understanding of   
data that they are confronted with.
An essential functionality towards this goal
is the capability to provide meaningful explanations about why
data is present or missing form the result of a query.
Explanations help data consumers gauge how much trust one
can place on the result.
Perhaps more importantly, they provide useful information for 
debugging the query or data that led to 
incorrect results. 

This is particularly the case in scenarios where complex data analysis tasks are 
specified through large collections of nested views 
(i.e., views that may be defined in terms of other views). 
For example, schemas with nested view definitions and integrity constraints
capture the core of LogiQL~\cite{LogiQL-PODS,LogiQL14,LogiQL}
(where view definitions may, in general, involve not only relational operations, but also
aggregation, machine learning and 
mathematical optimization tasks). LogiQL is
a language developed and used at LogicBlox \cite{LogicBlox-SIGMOD}
for developing data-intensive ``self service'' applications involving
 complex data analytics workflows. 
 Similar recent industrial systems include 
 Datomic\footnote{\url{www.datomic.com}} and Google's
Yedalog~\cite{Yedalog}.
In each of these systems, nested view definitions (or, Datalog programs)
are used to specify complex workflows
to drive data-analytics tasks. Explanations for unexpected query results
(such as an unexpected tuple or a missing tuple) 
are very useful in such settings, since the source of an error can 
be particularly hard to track. 

There has been considerable research on the topic of 
deriving explanations for why a tuple belongs to the output
of a query.
Early systems were developed in~\cite{ARRSS93, ST90} to 
provide explanations for answers to logic programs in the context
of a deductive database.
The presence of a tuple in the output is explained by 
enumerating 
all possible derivations, that is, instantiations of the logic rules that derive the answer tuple.
In~\cite{ST90}, the system also explains missing answers, by providing  a 
partially instantiated rule, 
based on the missing tuple, and leaving the user to figure out 
how the rest of the rule would have to be instantiated.
In the last decade or so, there has been significant efforts to 
characterize different notions of provenance (or lineage) of query 
answers (see, e.g.,~\cite{CCT09, semirings}) which 
can also be applied to understand why an answer is in the query result.\looseness=-1

There have
also been extensive studies on the {\em why-not problem} (e.g., more recent studies include \cite{BWSDN15, CJ09, HCDN08, HHT09, MGMS11, TC10}). The why-not problem is
the problem of explaining why an answer is missing 
from the output.
Since~\cite{ST90}, the {\em why-not problem} 
was also studied in~\cite{HHT09, HCDN08} in the context of 
debugging results of data extracted via select-project-join queries, and, subsequently, 
a larger class of queries that also includes union and aggregation operators.
Unlike~\cite{ST90} which is geared towards providing explanations for answers
and missing answers,
the goal in~\cite{HCDN08} is to propose modifications to underlying
database $I$, yielding another database $I'$ based on the provenance of
the missing tuple, constraints, and trust specification at hand, so that 
the missing tuple appears in the result of the same query 
$q$ over the updated database $I'$.
In contrast to the {\em data-centric} approach of updating the 
database to derive the missing answer, another 
line of research~\cite{BHT14,CJ09,TC10} 
follows a {\em query-centric} approach whereby 
the query $q$ at hand is modified to $q'$ (without modifying the underlying database)
so that the missing answer appears in the output of $q'(I)$.

\smallskip\noindent\textbf{A new take on why-not questions: }
In this paper, we develop a novel foundational framework 
for why-not explanations
that is principally different from prior approaches.
Our approach 
is neither
data-centric nor query-centric. Instead, we derive high-level explanations via 
an ontology that is either provided, or is derived from the data or schema. 
Our immediate goal is not to compute repairs of the underlying database or query so 
that the missing answer would
appear in the result.
Rather, as in~\cite{ST90}, our primary goal is to provide understandable explanations
for why an answer is missing from the query result. 
As we will illustrate, 
explanations that are based on an ontology  have the potential
to be high-level and provide meaningful insight to why a tuple is
missing from the result.  
This is because an ontology abstracts a domain in terms of concepts and relationships
amongst concepts. Hence, explanations that are based on concepts and relationships
from an ontology 
will embody such high-level abstractions. 
As we shall describe, our work considers two cases. The first is 
when an ontology is provided
externally, in which case explanations will embody external knowledge about the domain.
The second is when
an ontology is not provided. For the latter, we allow an ontology to be derived
from the schema, and hence
explanations will embody knowledge about the domain through concepts and relationships
that are defined
over the schema.\looseness=-1

Formally, an explanation for why
a tuple $\tuple{a}$ is not among the results of a query $q(I)$,  in
our framework, is a tuple of concepts from the ontology whose
extension includes the missing tuple $\tuple{a}$ and, at the same
time, does not include any tuples from $q(I)$.
For example, a query may ask for all products that each store has in
stock, in the form of (product ID, store ID) pairs, from the database
of a large retail company.
A user may then ask why is the pair \emph{(P0034, S012)} not among the
result of the query.
Suppose {\em P0034} refers to a bluetooth headset product and
{\em S012} refers to a particular store in San Francisco.
If {\em P0034} is an instance of a concept {\em bluetooth headsets} 
and {\em S012} is an instance of a concept {\em stores in San Francisco},
and suppose that no pair $(x,y)$, where $x$ is an instance of {\em bluetooth headset} 
and $y$ is an instance of {\em stores in San Francisco}, belongs to the 
query result.
Then the pair of concepts ({\em bluetooth headset}, {\em stores in San Francisco})
is an explanation for the given why-not question. Intuitively, it signifies
the fact that 
``{\em none of the stores in San Francisco has any bluetooth headsets on stock}''.

There may be
multiple explanations for a given why-not question.
In the above example, this would be the case if, for instance, 
{\em S012} belongs also to a more general concept {\em stores in California},
and that none of the stores in California have bluetooth headsets on stock.
Our goal is to
compute a {\em most-general explanation}, that is, an explanation
that is not strictly subsumed by any other explanation.
We study the complexity of computing a most-general explanation
to a why-not question. Formally, we define a  
{\em why-not instance} (or, {\em why-not question}) to be 
a quintuple $(\sch, I, q, \mbox{\em Ans}, \tuple{a})$ 
where
$\sch$ is a {\em schema}, which may include
integrity constraints; $I$ is an instance of $\sch$; $q$ is a query over $\sch$; {\em Ans} = $q(I)$; and
$\tuple{a}\not\in q(I)$.

As mentioned earlier, 
a particular scenario where why-not questions easily arise is when querying
schemas that include a large collection of views, 
and where each view may be nested, that is, 
defined in terms of other views. 
Our framework captures this setting, since view definitions 
can be expressed by means of constraints. 

Our framework supports a very general notion of an ontology, which we call
{\em $\sch$-ontologies}. For a given relational schema $\sch$, an $\sch$-ontology is 
a triple $(\concepts, \sqsubseteq, ext)$ which defines the
set of concepts, the subsumption relationship between concepts, and
respectively, the extension of each concept w.r.t. an instance of the
schema $\sch$. 
We use this general notion of an $\sch$-ontology 
to formalize the key notions of {\em explanations} and {\em most-general explanations},
and we show that $\sch$-ontologies capture 
two different types of ontologies.

The first type of ontologies we consider are 
those that are defined externally, provided that there is a way to associate
the concepts in the externally defined ontology to the instance at hand.
For example, the ontology may be represented in the form of a
Ontology-Based Data Access (OBDA) specification~\cite{PLCDL2008}.
More precisely, an OBDA 
specification consists of a set of concepts and subsumption relation 
specified by means of a description logic 
terminology,
and a set of {\em mapping assertions} that relates the concepts to a relational 
database schema at hand.
Every OBDA specification induces a corresponding $\sch$-ontology.
If the
concepts and subsumption relation are defined by a TBox in a tractable description 
logic such as \dlliteshort, and  the mapping assertions
are Global-As-View (GAV) assertions,  the induced $\sch$-ontology can in fact
be computed from the OBDA specification in polynomial time.
We then present an
algorithm for computing all most-general explanations to a
why-not question, given  an external $\sch$-ontology.
The algorithm 
runs in polynomial time when the arity of the query is bounded, 
and it executes in exponential time in general.
We show that the exponential running time
is unavoidable, unless P$=$NP, 
because
the problem of deciding whether or not there exists an explanation 
to a why-not question given an external $\sch$-ontology is \NP-complete in general.

The second type of ontologies that we 
consider are ontologies
that are derived either (a) from
a  schema $\sch$, or (b) from an instance of the schema. 
In both cases, the concepts of the ontology 
are defined through concept expressions in a suitable language $L_\sch$ that
we develop. Specifically, our concepts are obtained from the relations in the schema,
through selections, projections, and intersection.
The difference between the two cases lies in the 
way the subsumption relation $\sqsubseteq$ is defined.
In the former, a concept $C$ is considered to be subsumed in another concept $C'$ 
if the extension of $C$ is contained in the extension of $C'$ over all instances
of the schema.  For the latter, subsumption is considered to hold
if the extension of $C$ is contained in the extension of $C'$ with respect to 
the given instance
of the schema.
The $\sch$-ontology induced by a schema $\sch$, or instance $I$,
denoted $\ontocal_\sch$ or $\ontocal_I$, respectively, is 
typically
infinite, and is not intended to be materialized. Instead, we present 
an algorithm for directly computing a most-general explanation with respect
to $\ontocal_I$. The algorithm runs in exponential time in general. However,
if the schema is of bounded arity, the algorithm runs in
polynomial time. 
As for computing most-general explanations with respect to $\ontocal_\sch$, 
we identify
 restrictions on the integrity constraints under
which the problem is decidable, and we present complexity upper bounds for
these cases. 

\smallskip
\noindent
{\bf More related work: }
The use of ontologies to facilitate access to databases is not new.
A prominent example is OBDA, where
queries are either posed directly against an ontology, or 
an ontology is used to enrich a data schema against which queries 
are posed with additional 
relations (namely, the concepts from the ontology) \cite{BtCLW, PLCDL2008}.
Answers are computed based on an open-world assumption and 
using the mapping assertions and ontology provided by the OBDA specification.
As we described above, we make use of OBDA specifications as a means to 
specify an external ontology and with a database instance through mapping assertions.
However, unlike in OBDA, we consider queries  posed against 
 a database instance under 
the traditional closed-world semantics, and the ontology is used only to 
derive why-not explanations.

The problems of providing why explanations and why-not explanations have also been
investigated in the context of OBDA 
in \cite{BorgidaCR08} and \cite{CalvaneseOSS13}, respectively.
The why-not explanations of~\cite{CalvaneseOSS13} follow the {\em data-centric} 
approach to why-not 
provenance as we discussed earlier where their goal is to modify the assertions that describe the 
extensions of concepts in the ontology so that the missing tuple will appear in the query result.

There has also been prior work on extracting ontologies from data.
For example, in \cite{lubyte}, the authors considered heuristics to automatically
generate
an ontology from a relational database by defining project-join 
queries over the data. 
Other examples on ontology extraction from data include
publishing relational data as RDF graphs or statements
(e.g., D2RQ~\cite{Bizer04}, Triplify~\cite{Aueretal2009}).
We emphasize that our goal is not to extract and materialize ontologies, but rather,
to use an ontology that is derived from data to compute why-not explanations.

\smallskip\noindent\textbf{Outline: }
After the preliminaries, in Section~\ref{sec:TheWhyNotProvenanceProblem}
we present our framework for why-not explanations. In 
Section~\ref{sec:ontology} we discuss in detail 
the two ways of obtaining an $\sch$-ontology.
In Section~\ref{sec:AlgorithmsForWhyNotExplanations} we present our
main algorithmic results. Finally, in Section~\ref{sec:VariationsOfTheFramework},
we study variatations of our framework, including the problem of 
producing \emph{short} most-general explanations, and alternative notions
of \emph{explanation}, and of what it means to be \emph{most general}.

\section{Preliminaries}
\label{sec:Preliminaries}

A \emph{schema} is a pair $(\sch, \Sigma)$, where $\sch$ is a set $\{R_1, \ldots, R_n\}$ of relation names,
where each relation name has an associated arity, and $\Sigma$ is a set of first-order sentences over $\sch$, which we will refer as \emph{integrity constraints}. Abusing the notation, we will write $\sch$ for the schema $(\sch, \Sigma)$.
A \emph{fact} 
is an expression of the form $R(b_1, \ldots, b_k)$ where $R\in\sch$ is a relation 
of arity $k$, and for $1\leq i\leq k$, we have $b_i\in \const$, where $\const$ is
a countably infinite set of constants. We assume a dense linear order $<$ on $\const$. 
An \emph{attribute} $A$ of an $k$-ary relation name $R\in\sch$ is
a number $i$ such that $1\leq i\leq k$. 
For a fact $R(\tuple{b})$ where $\tuple{b}=b_1, \ldots, b_k$, we sometimes
write
$\pi_{A_1,...,A_k}(\tuple{b})$ to mean the tuple consisting of the
$A_1$th, ..., $A_k$th constants in the tuple $\tuple{b}$, that is, 
the value 
$(b_{A_1}, \ldots, b_{A_k})$.
An \emph{atom over $\sch$} is an expression $R(x_1, \ldots, x_n)$, where $R\in \sch$ and every $x_i, i\in \{1, \ldots, n\}$ is a variable or a constant. 

A \emph{database instance}, or simply an \emph{instance}, $I$ over
$\sch$ is a set of facts over $\sch$ satisfying the integrity constraints $\Sigma$. Equivalently, an instance $I$ is
a map that assigns to each $k$-ary relation name $R\in\sch$ a finite
set of $k$-tuples over $\const$ such that the integrity constraints are satisfied. By $R^{I}$ we denote the set of these tuples. 
We write $\inst{\sch}$ to 
denote the set of all database instances over $\sch$, and $\adom{I}$
to denote the active domain of $I$, i.e., the set of all constants
occurring in facts of $I$.

\smallskip
\noindent
{\bf Queries~}
A \emph{conjunctive query} (CQ) over $\sch$ is a query of the form
$\exists \tuple{y}.\phi(\tuple{x},\tuple{y})$ where
$\phi$ is a conjunction of atoms over $\sch$.
Given an instance $I$ and a CQ $q$, we write $q(I)$ to denote 
the set of answers of $q$ over $I$. 
In this paper we allow conjunctive queries containing comparisons to constants, that is, comparisons of the form $x\,\mathtt{op}\,c$, where $\mathtt{op}\in\{=, <, > , \leq, \geq\}$ and $c\in \const$. We show that all upper bounds hold for the case of CQs with such comparisons, and all lower bounds hold without the use of comparisons (unless explicitly specified otherwise). We do \emph{not} allow comparisons between variables. 

\smallskip
\noindent {\bf Integrity constraints~}
In this paper we consider different classes of integrity 
constraints, including functional dependencies and inclusion dependencies.
We also consider $\UCQ$-view definitions and nested $\UCQ$-view definitions,
which can be expressed using integrity constraints as well. 

A \emph{functional dependency} (FD) on a relation $R\in \sch$ is an expression of the 
form $R:\ X\to Y$
where $X$ and $Y$ are subsets of the set of attributes of $R$. 
We say that an instance $I$ over $\sch$ satisfies the FD if 
for every $\tuple{a}_1$ and $\tuple{a}_2$ from $R^I$ if $\pi_A(\tuple{a}_1) = \pi_A(\tuple{a}_2)$ for every $A\in X$, then $\pi_B(\tuple{a}_1)=\pi_B(\tuple{a}_2)$ for every $B\in Y$.

An \emph{inclusion dependency} (ID) is an expression of the form 
\[
R[A_1, \ldots, A_n] \subseteq S[B_1, \ldots, B_n]
\]
where $R, S\in \sch$, each $A_i$ and $B_j$ is an attribute of $R$ and $S$ respectively. 
We say that an instance $I$ over $\sch$ satisfies the ID if 
\[
\{\pi_{A_1,...,A_n}(\tuple{a}) \mid \tuple{a}\in R^I\} \subseteq \{\pi_{B_1,...,B_n}(\tuple{b}) \mid \tuple{b}\in S^I\}. 
\]

Note that functional and integrity constraints can equivalently be written as first-order sentences~\cite{abit:95}.

\smallskip
\noindent{\bf View Definitions}
To simplify presentation, we treat view defintions as a 
special case of integrity constraints. 

A set of integrity constraints $\Sigma$ over $\sch$ is said to be a \emph{collection of $\UCQ$-view definitions} if there exists a partition $\sch = \dsch \cup \vsch$ such that for every $P\in \vsch$, $\Sigma$ contains exactly one first-order sentence of the form:
\[
P(\bar{x}) \leftrightarrow \vee^k_{i=1} \phi_i (\bar{x}), \eqno(*)
\]
where each $\phi_i$ is a conjunctive query (with comparisons to constants) over $\dsch$. 

Similarly, a set of integrity constraints $\Sigma$ over $\sch$ is said to be a \emph{collection of nested $\UCQ$-view definitions} if there exists a partition $\sch = \dsch \cup \vsch$ such that for every $P\in \vsch$, $\Sigma$ contains exactly one first-order sentence of the form (*), where each $\phi_i$ is now allowed to be a conjunctive query over $\dsch\cup\vsch$, but subject to the following acyclicity condition. Let us say that $P\in \vsch$ \emph{depends on} $R\in \vsch$, if $R$ occurs in the view definition of $P$, that is, in the sentence of $\Sigma$ that is of the form (*) for $P$. We require that the ``depends on'' relation is acyclic.
If, in the view definition of every $P\in\vsch$, each disjunct $\phi_i$ contains at most one atom over $\vsch$,  then we say that $\Sigma$ is a collection of \emph{linearly} nested $\UCQ$-view definitions.

Note that a collection of nested $\UCQ$-view definitions (in the absence of comparisons) can be equivalently viewed as a non-recursive Datalog program and vice versa~\cite{bene:impa10}. In particular, a collection of linearly nested $\UCQ$-view definitions corresponds to a linear non-recursive Datalog program.

\begin{example}
\label{ex:dws}
As an example of a schema, consider $\sch = \dsch \cup \vsch$ with the integrity constraints in Figure~\ref{fig:example-dataworkspace}.
An instance $I$ of the schema $\sch$ is given in Figure~\ref{fig:train-instance}.
\hfill$\Box$
\end{example}

\begin{figure}[t]
\scriptsize
\par\noindent

$\begin{array}{@{}lll@{}}
\multicolumn{3}{@{}l} {\text{Data schema } \dsch}: \\  \\
\{ \text{Cities(name, population, country, continent)}, \\
\text{ \ \ Train-Connections(city\_{}from, city\_{}to)} \} \\ 
\end{array}$

\medskip

$\begin{array}{@{}lll@{}}
\multicolumn{3}{@{}l} 
{\text{View schema } \vsch}: \\ \\
\{ \text{BigCity(name), EuropeanCountry(name)}, \\
\text{ \ \ Reachable(city\_from, city\_to)}\}\\
\end{array}$

\medskip

\begin{tabular}{@{}lll@{}}
$\UCQ$-view definitions: \\ \\
BigCity($x$) &  $\leftrightarrow$ & Cities($x$,$y$,$z$,$w$) $\wedge$ $y \geq 5000000$ \\
EuropeanCountry($z$)& $\leftrightarrow$  & Cities($x$,$y$,$ z$,$w$) $\wedge$ $w = \text{Europe}$ \\
Reachable($x$,$y$) & $\leftrightarrow$  & Train-Connections($x$,$y$)  $\vee$ \\
& & (Train-Connections($x$,$z$) $\wedge $ Train-Connections ($z$,$y$)) \\
\end{tabular}

\medskip

$\begin{array}{@{}lll@{}}
\multicolumn{3}{@{}l}{\text{Functional and inclusion dependencies:}} \\ \\
\text{country}& \rightarrow & \text{ continent }\\
\text{BigCity[name]} & \subseteq & \text{Train-Connections[city\_from]} \\
\text{Train-Connections[city\_from]} & \subseteq & \text{Cities[name]} \\
\text{Train-Connections[city\_to]} & \subseteq &\text{Cities[name]} 
\end{array}$

\caption{Example of a schema $\sch$.}
\label{fig:example-dataworkspace}
\end{figure}

\begin{figure}[t]
\scriptsize
\begin{tabular}{@{}ll@{}}
\textbf{Cities} & \textbf{Train-Connections}\\
\begin{tabular}[t]{|@{\,}l@{\,}|@{\,}l@{\,}|@{\,}l@{\,}|@{\,}l@{\,}|}
 \hline
   \textbf{name} & \textbf{population} & \textbf{country} & \textbf{continent} \\
	\hline
	\hline
   Amsterdam & 779,808 & Netherlands & Europe \\
   Berlin & 3,502,000 & Germany & Europe \\
   Rome & 2,753,000 & Italy & Europe \\
   New York & 8,337,000 & USA & N.America \\
   San Francisco & 837,442 & USA & N.America \\
   Santa Cruz & 59,946 & USA & N.America \\
   Tokyo & 13,185, 000 &Japan & Asia \\
   Kyoto & 1,400,000 & Japan & Asia \\
		 \hline
\end{tabular} &

\begin{tabular}[t]{|@{\,}l@{\,}|@{\,}l@{\,}|}
 \hline
	\textbf{city\_{}from} & \textbf{city\_{}to} \\
	\hline
	\hline
	Amsterdam & Berlin \\
	Berlin & Rome \\
	Berlin & Amsterdam \\
	New York & San Francisco \\
	San Francisco & Santa Cruz \\
	Tokyo & Kyoto \\
	 \hline
\end{tabular}
\end{tabular}

\bigskip
\begin{tabular}{@{}ll@{               \ \ }ll@{}}
\textbf{BigCity} & \textbf{EuropeanCountry} & \textbf{Reachable}\\
\begin{tabular}[t]{|@{\,}l@{\,}|}
 \hline
	\textbf{name}  \\
	\hline
	\hline
	New York \\
	Tokyo \\
	 \hline
\end{tabular}&
\begin{tabular}[t]{|@{\,}l@{\,}|}
 \hline
	\textbf{name}  \\
	\hline
	\hline
	 Netherlands \\
	 Germany \\
	Italy \\
	 \hline
\end{tabular} &

\begin{tabular}[t]{|@{\,}l@{\,}|@{\,}l@{\,}|}
 \hline
	\textbf{city\_{}from} & \textbf{city\_{}to} \\
	\hline
	\hline
	Amsterdam & Berlin \\
	Berlin & Rome \\
	Berlin & Amsterdam \\
	New York & San Francisco \\
	San Francisco & Santa Cruz \\
	Tokyo & Kyoto \\
	Amsterdam & Rome \\
	Amsterdam & Amsterdam \\
	Berlin & Berlin \\
	New York & Santa Cruz \\
	 \hline
\end{tabular}
\end{tabular}
\caption{Example of an instance $I$ of $\sch$.
} 
\label{fig:train-instance}
\end{figure}

\section{Why-Not Explanations}
\label{sec:TheWhyNotProvenanceProblem}

\label{sec:Framework}

Next, we introduce our ontology-based framework for explaining why a 
tuple is not in the output of a query.
Our framework is based on a general
notion of an ontology.
As we shall describe in Section~\ref{sec:ontology}, 
the ontology that is used may be an external ontology (for example, 
an existing ontology specified in a description logic),
or it may be an ontology that is derived from 
a schema.
Both are a special case of 
our general definition of an $\sch$-ontology.

\begin{definition} [$\sch$-ontology] An {\em $\sch$-ontology} over a relational schema 
$\sch$ is a triple $\mathcal{O}=\sontoshort$, where
\begin{itemize}
\item $\concepts$ is a possibly infinite set, whose elements are called {\em concepts},
\item $\sqsubseteq$ is a pre-order (i.e., a reflexive and transitive binary relation) on $\concepts$, called the {\em subsumption relation}, and
\item $ext: \concepts \times \inst{\sch} \rightarrow \wp(\const)$ is a polynomial-time computable function that 
will be used to identify instances of a concept in a given database 
instance ($\wp(\const)$ denotes the powerset of $\const$). 
\end{itemize} 
More precisely, we assume that $ext$ is specified by a Turing machine that,
given $C\in \concepts$, $I\in \inst{\sch}$ and $c\in\const$, 
decides in polynomial time if $c\in ext(C,I)$.

A database instance $I\in\inst{\sch}$ is \emph{consistent} with $\cal O$
if, for all $C_1, C_2\in \concepts$ with $C_1 \sqsubseteq C_2$, we
have  $ext(C_1, I) \subseteq ext(C_2, I)$. \looseness=-1
\end{definition}

An example of an $\sch$-ontology $\ontocal = \sontoshort$ is 
shown in Figure~\ref{fig:toy-ontology}, where the concept subsumption
relation $\sqsubseteq$ is depicted by means of a Hasse diagram. 
Note that, in this example, 
$ext(C,I)$ is independent of the database instance $I$
(and, as a consequence, every $\sch$-instance is consistent with $\cal O$). 
In general,
this is not the case (for example, the extension of a concept may be
determined through mapping assertions,
cf.~Section~\ref{sec:OntologyBasedDataAccessScenario}).

We define our notion of an ontology-based explanation next.

\begin{definition}[Explanation]
\label{dfn:explanation}
Let ${\cal O}=\sontoshort$ be an $\sch$-ontology, $I$ an $\sch$-instance consistent with ${\cal O}$.
Let 
$q$ be an $m$-ary query over $\sch$, 
and $\tuple{a}=\tup$ a tuple of constants 
such that $\tuple{a}\not\in q(I)$. 
Then a tuple of concepts $(C_1, \ldots, C_m)$ from ${\cal C}^m$ is called 
an \emph{explanation} for  $\tuple{a}\not \in q(I)$ with respect to $\ontocal$ (or
an {\em explanation} in short) if:
\begin{itemize}
\item for every $1\leq i\leq m$, $a_i\in \exs{C_i}{I}$, and
\item $(\exs{C_1}{I}\times \ldots \times \exs{C_m}{I})\cap q(I) =\emptyset$.
\end{itemize}
\end{definition}

In other words, an explanation is a tuple of concepts whose 
extension includes the missing tuple $\tuple{a}$ (and thus explains $\tuple{a}$)
but, at the same time, it does not include any tuple in $q(I)$ (and thus
does not explain any tuple in $q(I)$). Intuitively, the 
tuple of concepts is an explanation 
that is orthogonal to existing tuples in $q(I)$ but 
relevant for the missing tuple $\tuple{a}$, and thus forms an explanation for
why $\tuple{a}$ is not in $q(I)$. There can be multiple explanations in general
and
the ``best'' explanations are the ones that are the most general.

\begin{definition}[Most-general explanation]
Let ${\cal O}=\sontoshort$ be an $\sch$-ontology, and let  
$E=(C_1, \ldots, C_m)$ and $E'=(C'_1, \ldots, C'_m)$ be two 
tuples of concepts from $\concepts^m$.
\begin{itemize}
\item
We say that $E$ is \emph{less general} than $E'$ with respect to $\cal O$, 
denoted as $E\leq_{\cal O} E'$, if $C_i\sqsubseteq C_i'$ for every $i, 1\leq i\leq m$.
\item
We say that $E$ is \emph{strictly less general} 
than $E'$ with respect to $\cal O$, denoted as $E<_{\cal O} E'$, 
if $E\leq_{\cal O} E'$, and $E' \not \leq_{\cal O} E$.
\item
We say that $E$ is a {\em most-general explanation} 
for $\tuple{a}\not\in q(I)$ if $E$ is an explanation 
for $\tuple{a}\not\in q(I)$, and there is no 
explanation $E'$ for $\tuple{a}\not\in q(I)$ such that $E'>_{\cal O} E$.
\end{itemize}
\end{definition}

As we will formally define in Section~\ref{sec:AlgorithmsForWhyNotExplanations}, 
a {\em why-not problem} asks the question: ``why is the tuple $\tup$ not 
in the output of a query $q$ over an instance $I$ of schema $\das$?''
The following example illustrates the notions of explanations and most-general explanations
in the context of a why-not problem.

\begin{example}
\label{ex:first-big-example}
Consider the instance $I_\dsch$ 
of the relational schema $\sch$ = \{Cities(name, population, country, continent),
Train-Connections(city\_{}from, city\_{}to)\}
shown in Figure~\ref{fig:train-instance}. 

Suppose $q$ is the query
$\exists z.$ Train-Connections($x,z$) $\wedge$ Train-Connections($z,y$). 
That is, the query asks for all pairs of cities that are 
connected via a city.
Then 
$q(I)$ returns tuples  $\{\tupp{\text{Amsterdam, Rome}}, \tupp{\text{Amsterdam, Amsterdam}},  \tupp{\text{Berlin, Berlin}},$
$ \tupp{\text{New York, Santa Cruz}}\}$. 
A user may ask why is the tuple
$\langle$Amsterdam, New York$\rangle$ not in the result of $q(I)$ 
(i.e., why is
$\langle$Amsterdam, New York$\rangle \not\in q(I)$?). 
Based on the $\sch$-ontology
defined in 
Figure~\ref{fig:toy-ontology},
we can derive the following 
explanations for $\langle$Amsterdam, New York$\rangle \notin q(I)$ :

\[\small\begin{array}{lll}
E_1 = \langle$Dutch-City, East-Coast-City$\rangle\\
E_2 = \langle$Dutch-City, US-City$\rangle\\
E_3 = \langle$European-City, East-Coast-City$\rangle\\
E_4 = \langle$European-City, US-City$\rangle
\end{array}\]
$E_1$ is the simplest explanation, i.e., the one we can build by looking at the lower level of the hierarchy in our $\sch$-ontology. Each subsequent explanation is more general than at least one of the prior explanations
w.r.t. to the $\sch$-ontology.
In particular, we have
$E_4 >_{\cal O} E_2 >_{\cal O} E_1$, and $E_4 >_{\cal O} E_3 >_{\cal O} E_1$. 
Thus, the most-general explanation for why 
$\langle$Amsterdam, New York$\rangle \not\in q(I)$ with respect to our $\sch$-ontology
is $E_4$, which intuitively informs that the reason is 
because Amsterdam is a city in Europe while New York is a city in the US (and hence, they
are not connected by train).
Note that all the other possible combinations of concepts are not explanations 
because they intersect with $q(I)$. 
\hfill$\Box$
\end{example}

\begin{figure}[t!]
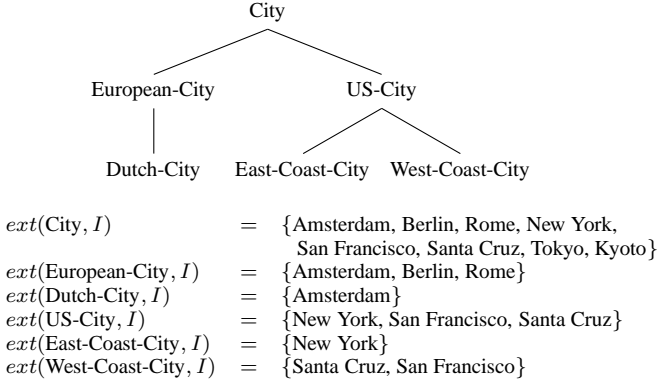

\begin{center}
{\small
\Tree [.City [.European-City Dutch-City ] [.US-City [.East-Coast-City ] [.West-Coast-City  ] ] ] 
}
 \end{center}
\[\small\begin{array}{lll}
 ext(\text{City},I) &=& \{\text{Amsterdam, Berlin, Rome, New York,}\\&&\text{~~ San Francisco, Santa Cruz, Tokyo, Kyoto}\} \\
 ext(\text{European-City},I) &=& \{\text{Amsterdam, Berlin, Rome}\} \\
 ext(\text{Dutch-City},I) &=& \{\text{Amsterdam}\} \\
 ext(\text{US-City},I) &=& \{\text{New York, San Francisco, Santa Cruz}\} \\
 ext(\text{East-Coast-City},I) &=& \{\text{New York}\} \\
 ext(\text{West-Coast-City},I) &=& \{\text{Santa Cruz, San Francisco}\} \\
 \end{array}\]
\vspace{-1mm}
\caption{Example ontology.}
\label{fig:toy-ontology}
\vspace{-3mm}
\end{figure}

As we will see in Example~\ref{ex:explanationsWithDerived},
there may be more than one most-general explanations in general.

Generalizing the above example, we can informally define the
problem of explaining why-not questions via ontologies as 
follows:
{\em
given an instance $I$ of schema $\sch$, a query $q$ over $\sch$, an $\sch$-ontology $\cal O$ (consistent with $I$)
and a tuple $\tuple{a}\not \in q(I)$, compute a most-general explanation
for $\tuple{a}\not \in q(I)$, if one exists, w.r.t. $\cal O$.
}
As we shall describe in Section~\ref{sec:AlgorithmsForWhyNotExplanations}, in addition to the 
above problem of 
computing one most-general explanation, we will also investigate the corresponding decision problem that asks whether or not an explanation for a why-not problem exists, and whether 
or not a given tuple of concepts is a most-general explanation for a why-not problem. 
In our framework, the $\sch$-ontology $\cal O$ may be given explicitly as part of the input,
or it may be derived from a given database instance or a given schema. We will introduce 
the different
scenarios by which an ontology may be obtained in the next section, before we 
describe our algorithms for computing most-general explanations in Section~\ref{sec:AlgorithmsForWhyNotExplanations}.

\section{Obtaining Ontologies}\label{sec:ontology}

In this section we discuss two approaches by which
$\sch$-ontologies may be obtained. 
The first approach allows one to 
leverage 
an external ontology, provided that there is a way to relate a concept in the
ontology to a database instance. 
In this case, the set $\concepts$ of concepts is specified through a description logic
such as $\cal ALC$ or \textit{DL-Lite}; 
$\sqsubseteq$ is a partial order on the concepts defined in the ontology, and the
function {\em ext} may be given through \emph{mapping assertions}.
The second approach considers an $\sch$-ontology that is 
derived from a specific database instance, or from
a schema. 
This approach is useful as it allows
one to define an ontology to be used for explaining why-not questions in the absence
of an external ontology.

In either case, we
study the complexity of deriving such $\sch$-ontologies 
based on the language on
which concepts are defined, the subsumption between concepts, and the function $ext$,
which is defined according to the semantics of the concept language.

\subsection{Leveraging an external ontology}
\label{sec:OntologyBasedDataAccessScenario}

We first consider the case where we are given an external ontology
that models the domain of the database instance, and a relationship
between the ontology and the instance.
We will illustrate in particular how \emph{description logic
ontologies} are captured as
a special case of our framework.  

In what follows, our exposition borrows notions from the
Ontology-Based Data Access (OBDA) framework. 
Specifically, we will make crucial use of the notion
of an \emph{OBDA specification} \cite{DiPinto2013}, which consists of
a description logic ontology, a relational schema, and a collection of
mapping assertions.  To keep the exposition simple, we restrict our
discussion to one particular description logic, called \dlliteshort, which
is a representative member of the \emph{DL-Lite} family of description
logics~\cite{calv2007}.
\dllite is the basis for the OWL 2
QL\footnote{\small\url{http://www.w3.org/TR/owl2-profiles/\#OWL\_2\_QL}} profile of  OWL 2, which is a standard
ontology language for Semantic Web adopted by W3C.
As the other languages in the
\emph{DL-Lite} family, \dllite exhibits a good trade off between expressivity and complexity bounds for
important reasoning tasks such as subsumption checking, instance
checking and query answering.  

\medskip\par\noindent\textbf{TBox and Mapping Assertions.}
In the description logic literature, an ontology is typically formalized as a TBox (Terminology Box),
which consists of finitely many \emph{TBox axioms}, where each
TBox axiom 
expresses a relationship between concepts. 
Alongside TBoxes, ABoxes (Assertion Boxes) are sometimes used to 
describe the extension of concepts.
To simplify  the presentation, we do not
consider ABoxes here. 

\begin{definition}[\dlliteshort]
Fix a finite set $\Phi_C$ of \emph{``atomic concepts''} and a finite set $\Phi_R$ of \emph{``atomic roles''}. 
\begin{itemize}
\item
The 
\emph{concept expressions} and \emph{role expressions} of  
\dllite are defined as follows:
\begin{center}\begin{tabular}{ll}
Basic concept expression: & $B ::= A \mid \exists R$ \\
Basic role expression: & $R ::=  P \mid  P^-$\\
Concept expressions: & $C ::= B \mid \neg B$ \\
Role expressions & $E ::= R \mid \neg R $  \\
\end{tabular}
\end{center}
 where $A\in \Phi_C$ and $P\in\Phi_R$. 
 Formally, a
 \emph{$(\Phi_C,\Phi_R)$-interpretation} $\mathcal{I}$ is a map
 that assigns to every atomic concept in $\Phi_C$ a unary relation
 over $\const$ and to every atomic role in $\Phi_R$ a binary
 relation over $\const$.  The  map $\mathcal{I}$  naturally extends to
 arbitrary concept expressions and role expressions:
 \begin{center}
 \small
 \begin{tabular}{@{}ll@{}}
 $\mathcal{I}(P^-) = \{(x,y)\mid(y,x)\in \mathcal{I}(P)\}$ & $\mathcal{I}(\exists P) = \pi_1(\mathcal{I}(P))$ \\
 $\mathcal{I}(\neg P) = \const^2\setminus \mathcal{I}(P)$ &
$\mathcal{I}(\neg A) = \const\setminus \mathcal{I}(A)$ 
\end{tabular}
\normalsize
\end{center}
 Observe that 
$\mathcal{I}(\exists P^-) = \pi_2(\mathcal{I}(P))$.

\item
 A \emph{TBox (Terminology Box)} is a finite set of \emph{TBox axioms}
 where each TBox axiom is an inclusion assertion of the form $B
 \sqsubseteq C$ or $R \sqsubseteq E$, where $B$ is a basic concept
 expression, $C$ is a concept expression, $R$ is a basic role
 expression and $E$ is a role expression.  An
 $(\Phi_C,\Phi_R)$-interpretation $\mathcal{I}$ \emph{satisfies} a
 TBox if for each axiom $X \sqsubseteq Y$, it holds $\mathcal{I}(X)
 \subseteq \mathcal{I}(Y)$.

\item
 For concept expressions $C_1, C_2$ and  a TBox $\mathcal{T}$, we say that \emph{$C_1$ is subsumed by $C_2$ relative to $\mathcal{T}$} (notation: $\mathcal{T}\models C_1\sqsubseteq C_2$) if, for all interpretations $\mathcal{I}$ satisfying $\mathcal{T}$, we have that $\mathcal{I}(C_1)\subseteq \mathcal{I}(C_2)$.
 \end{itemize}
\end{definition}

\begin{figure}
\scriptsize
\par\noindent
\begin{tabular} {@{}ll@{}}
DL-Lite TBox axiom & (first-order translation) \\ \\
EU-City $\sqsubseteq$ City                   & $\forall x ~ \text{EU-City}(x)\to\text{City}(x)$ \\
Dutch-City $\sqsubseteq$ EU-City  & $\forall x~ \text{Dutch-City}(x) \to \text{EU-City}(x)$ \\
N.A.-City $\sqsubseteq$ City               & $\forall x ~ \text{N.A.-City}(x)\to\text{City}(x)$ \\
EU-City $\sqsubseteq \neg$ N.A.-City       & $\forall x ~ \text{EU-City}(x)\to\neg\text{N.A.-City}(x)$ \\
US-City  $\sqsubseteq $ N.A.-City       & $\forall x ~ \text{US-City}(x)\to\text{N.A.-City}(x)$ \\
City $\sqsubseteq \exists$ hasCountry        & $\forall x ~ \text{City}(x)\to\exists y ~ \text{hasCountry}(x,y)$\\
Country $\sqsubseteq \exists$ hasContinent   & $\forall x ~ \text{Country}(x)\to\exists y ~ \text{hasContinent}(x,y)$ \\
$\exists$hasCountry$^-$ $\sqsubseteq$ Country & $\forall x ~ (\exists y ~\text{hasCountry}(y,x)) \to \text{Country}(x) $  \\
$\exists$hasContinent$^-$ $\sqsubseteq$ Continent & $\forall x ~ (\exists y ~\text{hasContinent}(y,x)) \to \text{Continent}(x)$  \\
$\exists$connected $\sqsubseteq$ City       & $\forall x ~ (\exists y ~\text{connected}(x,y)) \to \text{City}(x)$ \\
$\exists$connected$^-$ $\sqsubseteq$ City   & $\forall x ~ (\exists y ~\text{connected}(y,x)) \to \text{City}(x)$ \\
\end{tabular}

\bigskip

$\begin{array}{@{}lll@{}}
\multicolumn{3}{@{}l}{\text{GAV mapping assertions (universal quantifiers omitted for readability):}} \\ \\
\text{Cities}(x,z,w,\text{``Europe''}) & \rightarrow & \text{EU-City($x$)}   \\
\text{Cities}(x,z,\text{``Netherlands''},w) & \rightarrow & \text{Dutch-City($x$)}   \\
\text{Cities}(x,z,w,\text{``N.America''}) & \rightarrow & \text{N.A.-City($x$)}\\
\text{Cities}(x,z,\text{``USA''},w) & \rightarrow & \text{US-City($x$)}   \\
\text{Cities}(x,y,z,w) & \rightarrow & \text{Continent($w$)} \\
\text{Cities}(x,k,y,w) & \rightarrow & \text{hasCountry(x,y)} \\
\text{Cities}(x,k,w,y) & \rightarrow & \text{hasContinent($x$,$y$)} \\
\text{Train-Connection}(x,y), & \\
~~\text{Cities}(x,x_1,x_2,x_3), \text{Cities}(y,y_1,y_2,y_3) & \rightarrow  & \text{connected($x$,$y$)} 
\end{array}$

\caption{Example DL-Lite ontology with mapping assertions.}
\label{fig:example-DL-Lite-ontology}
\end{figure}

An example of a \dllite TBox is given at the top of Figure~\ref{fig:example-DL-Lite-ontology}. 
For convenience, we have listed next to each TBox axiom, its equivalent semantics in first-order notation.

Next we describe what mapping assertions are.
Given an ontology and a relational schema, we can specify mapping assertions to relate
the ontology language to the relational schema, which is similar to how mappings
are used in OBDA \cite{PLCDL2008}.  
In general, mapping assertions are first order sentences over the schema
$\sch\cup\Phi_C\cup\Phi_R$ that express relationships between the
symbols in $\sch$ and those in $\Phi_C$ and $\Phi_R$.  
Among the different
schema mapping languages that can be used, we
restrict
our attention, for simplicity, 
to the class of {\em Global-As-View (GAV) mapping assertions} ({\em GAV mapping assertions or 
GAV constraints} or \emph{GAV source-to-target tgds}). 

\begin{definition}[GAV mapping assertions]
A GAV mapping assertion over $(\sch, (\Phi_C\cup\Phi_R))$ is a first-order sentence $\psi$ of the form 
\[ \forall \vec{x}\,\, (\phi_1(\vec{x_1}),\cdots,\phi_n(\vec{x_n}))\to\psi(\vec{x})\]
where $\vec{x} \subseteq \vec{x_1} \cup \ldots \cup \vec{x_n}$, 
$\phi_1, \ldots, \phi_n$ are atoms over $\sch$ and $\psi$ is an atomic formula
of the form $A(x_i)$ (for $A\in\Phi_C$) or $P(x_i,x_j)$ (for $P\in\Phi_R$). 
Let  $I$ be an $\sch$-instance and $\mathcal{I}$ an $(\Phi_C,\Phi_R)$-interpretation. 
We say that the pair $(I, \mathcal{I})$ \emph{satisfies} the GAV mapping assertion (notation: $(I,\mathcal{I})\models\psi$) if it holds that for any tuple of elements $\bar{a}$ from $adom(I)$, with $\bar{a}=\bigcup_{1\leq k\leq n}\bar{a_k}$, if $I\models \phi_1(\bar{a_1}),\ldots,\phi_n(\bar{a_n})$, then $a_i \in \mathcal{I} (A)$, with $a_i\in\bar{a}$ (if $\psi = A(x_i)$) or $(a_i, a_j)\in \mathcal{I}(P)$, with $a_i,a_j\in\bar{a}$ (if $\psi = P(x_i, x_j)$).
\end{definition}

Intuitively, a GAV mapping assertion associates a conjunctive query over $\bf S$ to 
an element (concept or atomic role) of the ontology.
A set of GAV mapping assertions associates, in general, a union of conjunctive queries
to an element of the ontology.
Examples of GAV mapping assertions are given at the bottom of 
Figure~\ref{fig:example-DL-Lite-ontology}. 

\medskip\par\noindent\textbf{OBDA induced ontologies}

\begin{definition}[OBDA specification]
Let $\mathcal{T}$ be a TBox, $\sch$ a relational schema,
and $\mathcal{M}$ a set of mapping assertions from $\sch$ to the concepts of $\mathcal{T}$. 
We call the triple $\obda = \bspec$ an {\em OBDA specification}. 

An $(\Phi_C,\Phi_R)$-interpretation $\mathcal{I}$ is said to be a \emph{solution} for 
an $\sch$-instance $I$ with respect to the OBDA specification $\obda$ if 
the pair $(I,\mathcal{I})$  satisfies all mapping assertions in $\mathcal{M}$ and $\mathcal{I}$ satisfies $\mathcal{T}$.
\end{definition}

Note that our notion of an OBDA specification is a special case of the
one given in \cite{DiPinto2013}, where we do not consider view
inclusion dependencies.
Also, as mentioned earlier, our OBDA specifications in this paper assume
that 
$\cal T$ is a \dllite TBox and $\cal M$ is a set of GAV mappings.
These restrictions allow us to achieve
good complexity bounds for explaining why-not questions with ontologies.
In particular, it is not hard to see that, for the OBDA specifications 
we consider, every $\sch$-instance $I$  has a solution.

\begin{restatable}{thm}{dllitecomp} \label{dl-lite-complexity}
(\cite{calv2007,PLCDL2008})
Let $\mathcal{T}$ be a \dllite TBox.
\begin{enumerate}
\item There is a \textsc{PTime}-algorithm for deciding subsumption. That is,
given $\mathcal{T}$ and two concepts $C_1,C_2$, decide if
$\mathcal{T} \models C_1 \sqsubseteq C_2$.
\item There is an algorithm that, given an OBDA specification $\mathcal{B}$, an instance $I$ over $\sch$ and a
     concept $C$, computes $certain(C,I,\mathcal{B}) = \bigcap\{\mathcal{I}(C) \mid \mathcal{I} \text{ is a solution for } I  \text{ w.r.t. } \mathcal{B}\}$.
     For a fixed OBDA specification, the algorithm runs in \PTIME (\ACz in data complexity).
\end{enumerate}
\end{restatable}

Every OBDA specification induces an $\sch$-ontology as follows.

\begin{definition}
Every OBDA specification $\obda = \bspec$ where $\mathcal{T}$ is a \dllite TBox and $\mathcal{M}$ is a set of GAV mappings gives rise to an $\sch$-ontology 
where:
\begin{itemize}
   \item ${\cal C}_{\ontocal_\obda}$ is the set of all basic concept expressions occurring in $\mathcal{T}$;
   \item $\sqsubseteq_{\ontocal_\obda} = \{(C_1,C_2)\mid \mathcal{T}\models C_1\sqsubseteq C_2\}$
   \item $ext_{\ontocal_\obda}$ is the polynomial-time computable function given by $ext_{\ontocal_\obda}(C,I) = \bigcap\{\mathcal{I}(C)\mid\text{$\mathcal{I}$ is a solution for $I$ w.r.t. $\mathcal{B}$}\}$ 
\end{itemize}
\end{definition}

Note that the fact that $ext_{\ontocal_\obda}$ is the polynomial-time computable follows from Theorem~\ref{dl-lite-complexity}.

We remarked earlier that, for the ODBA specifications $\obda$ that we consider, it holds
that every input instance has a solution. It follows that every input instance $I$
is consistent with the corresponding $\sch$-ontology $\ontocal_\obda$. 

\begin{restatable}{thm}{dllitegavcomp} \label{dl-lite-gav-complexity}The $\sch$-ontology
$\ontocal_\obda = ({\cal C}_{\ontocal_\obda},\sqsubseteq_{\ontocal_\obda},ext_{\ontocal_\obda})$ can be computed from a given 
OBDA specification $\obda = \bspec$ in \PTIME if $\mathcal{T}$ is a \dllite TBox and $\mathcal{M}$ is a set of GAV mappings.
 \end{restatable}

We are now ready to illustrate an example where a why-not question is explained
via an 
external ontology.

\begin{example}\label{example-obda}
Consider the OBDA specification $\obda =
(\mathcal{T},\sch,\mathcal{M})$ where $\mathcal{T}$ is the TBox
consisting of the \dllite axioms given in
Figure~\ref{fig:example-DL-Lite-ontology}, $\sch$ is the schema from
Example~\ref{ex:first-big-example}, and $\mathcal{M}$ is the set of
mapping assertions given in Figure~\ref{fig:example-DL-Lite-ontology}.
These together induce an $\sch$-ontology $\ontocal_\obda =
({\cal C}_{\ontocal_\obda},\sqsubseteq_{\ontocal_\obda},ext_{\ontocal_\obda})$.
The set ${\cal C}_{\ontocal_\obda}$ consists of the following basic concept expressions:

\smallskip
\small
City, EU-City, N.A.-City, Dutch-City,  \\
\indent
US-City, Country, Continent,\\
\indent
$\exists$ hasCountry, $\exists$ hasCountry$^-$, $\exists$ hasContinent, \\
\indent
$\exists$ hasContinent$^-$, $\exists$ connected, $\exists$ connected$^-$.
\normalsize
\smallskip

The set $\sqsubseteq_{\ontocal_\obda}$ includes the pairs of concepts of the TBox $\cal T$ given in 
Figure~\ref{fig:example-DL-Lite-ontology}.
We use the mappings
to compute the extension of each concept in ${\cal C}_{\ontocal_\obda}$ using the instance $I$
on the left of  Figure~\ref{fig:train-instance}. We list a few extensions here:
\[\small\begin{array}{lll}
 ext_{\ontocal_\obda}(\text{City},I) &=& \{\text{Amsterdam, Berlin, Rome, New York,}\\&&\text{ San Francisco, Santa Cruz, Tokyo, Kyoto}\} \\
 ext_{\ontocal_\obda}(\text{EU-City},I) &=& \{\text{Amsterdam, Berlin, Rome}\} \\
 ext_{\ontocal_\obda}(\text{N.A.-City},I) &=& \{\text{New York, San Francisco, Santa Cruz}\} \\
 ext_{\ontocal_\obda}(\exists \text{hasCountry$^-$},I) &=& \{\text{Netherlands, Germany, Italy, USA, Japan}\} \\
 ext_{\ontocal_\obda}(\exists \text{connected},I) &=& \{\text{Amsterdam, Berlin, New York}\} \\
\end{array}\]

Now consider the query 
$q(x,y)$ = $\exists z.$ Train-Connections($x,z$) $\wedge$ Train-Connections($z,y$), 
and $q(I)$ as in 
Example~\ref{ex:first-big-example}. 
As before, we would like to explain why is 
$\langle$Amsterdam, New York$\rangle \not\in q(I)$. 
This time, we use the induced $\sch$-ontology $\ontocal_\obda$ described
above to derive explanations 
for $\langle$Amsterdam, New York$\rangle \notin q(I)$:

\smallskip
\small
\centering
\begin{tabular}{ll@{~~~~~~~~~~}ll}
$E_1$ &= $\langle$EU-City, N.A.-City$\rangle$ & 
$E_2$ &= $\langle$Dutch-City,N.A.-City$\rangle$ \\ 
$E_3$ &= $\langle$EU-City, US-City$\rangle$ &
$E_4$ &= $\langle$Dutch-City,US-City$\rangle$ \\
\end{tabular}
\normalsize

\smallskip

Among the four explanations above, $E_1$  is the most general.
\hfill$\Box$
\end{example}

\subsection{Ontologies derived from a schema}
\label{DerivingOntology}

We now move to the second approach 
where an ontology is derived from an instance or a
schema. 
The ability to derive an ontology through an instance or a schema
is useful in the context where
an external ontology is unavailable.
To this purpose we first introduce a simple but suitable concept language
that can be defined over the schema $\sch$. 

Specifically, our concept language, denoted
as 
$L_\sch$, makes use of 
two relational algebra operations, projection ($\pi$) and selection ($\sigma$).  
We first introduce and motivate the language.
We will then describe our complexity results for testing
whether one concept is subsumed by another, and 
for obtaining an ontology from a given instance or a schema.
We will make use of these results later on in Section~\ref{sec:OntologyDerivedFromWorkspaceInstance}
and Section~\ref{sec:OntologyDerivedFromTheWorkspaceSchema}.

\begin{definition}[The Concept Language $L_\sch$]
Let $\sch$ be a schema.
A concept in $L_\sch$ is an expression $C$ defined by the following grammar. 
\begin{center}
$D ::= R \ \mid \ \sigma_{A_1 \mathtt{op}\, c_1, \ldots, A_n \mathtt{op}\, c_n}(R)$\\
\indent
$C := \top \mid \nom{c} \mid \pi_A (D)\  \mid \ C\sqcap C$
\end{center}
In the above, $R$ is a predicate name from $\sch$, $A, A_1,\ldots,A_n$ are 
attributes in $R$, not necessarily distinct, $c, c_1,\ldots,c_n\in\const$, and each occurrence 
of $\mathtt{op}$ is a comparison operator belonging to $\{=, < , > , \leq, \geq \}$. 
For $\mathbf{C}=\{C_1 , \ldots, C_k\}$  a finite set of concepts, 
we denote by $\sqcap \mathbf{C}$ the 
conjunction $C_1 \sqcap \ldots \sqcap C_k$. If $\mathbf{C}$ is empty, we 
take $\sqcap \mathbf{C}$ to be $\top$.

Given a finite set of constants $\mathcal{K}
\subset \const $, we define $L_\das[\calK]$ as
the concept language $L_\das$ whose concept expressions only use
constants from $\mathcal{K}$.  
By {\em selection-free $L_\das$}, we mean the language $L_\das$ where $\sigma$ is not allowed.
Similarly, by {\em intersection-free $L_\das$}, we mean the language $L_\das$ where $\sqcap$
is not allowed, and by $\LWmin$, we mean the minimal concept language
$L_\das$ where both $\sigma$ and $\sqcap$ are not allowed.
\end{definition}

Observe that the $L_\das$ grammar defines a concept in 
the form $C_1 \sqcap \ldots \sqcap C_n$ where each $C_i$ is $\top$ or $\nom{c}$ or
$\pi_A(R)$ or 
$\pi_A(\sigma_{A_1 \mathtt{op}\, c_1, \ldots, A_n \mathtt{op}\, c_n}(R))$.
A concept of the form $\{c\}$ is called a \emph{nominal}. A nominal $\{c\}$
is the ``most specific'' concept for the constant $c$.
Given a tuple $\tuple{a}$ that is not in the output, the corresponding tuple of 
nominal concepts forms a default, albeit trivial, 
explanation for why not $\tuple{a}$.

As our next example illustrates, even though our concept language $L_\das$
appears simple,
it is able to naturally capture many intuitive concepts over 
the domain of the database.

\begin{example}
\label{ex:Lsch-example}
We refer back to our schema $\das$ 
in Figure~\ref{fig:example-dataworkspace}.
Suppose we do not have access to 
an external ontology such as the one given in 
Example~\ref{ex:first-big-example}. 
We show that even so,
we can still construct meaningful concepts 
directly from the database schema
using the concept language described above. 
We list a few semantic concepts that can be specified with $L_\das$ in 
Figure~\ref{fig:concept-examples}, where we also show
the corresponding \textsc{select-from-where} style expressions and intuitive meaning.
\hfill$\Box$
\end{example}

Example~\ref{ex:Lsch-example} shows that, 
even though
$L_\sch$ is a simple language where concepts are 
essentially intersections of unary projections of relations and nominals, it is already
sufficiently expressive to capture natural concepts 
that can be used to build meaningful
explanations. It is worth noting that, for minor extensions of the language
$L_\sch$, such as with $\neq$-comparisons 
and
disjunction, the notion of a most-general explanation 
becomes trivial, in the sense that,
for each why-not question, 
there is a most-general explanation that 
essentially enumerates all tuples in the query answer.

By using $L_\sch$, we are able to define 
an ontology whose atomic concepts are derived
from the schema itself.
This approach allows us to provide explanations
using a vocabulary that is already familiar to the user.
We believe that this leads to intuitive and useful why-not
explanations.

If we view each expression $\pi_A(D)$ as an atomic concept, then 
the language $L_\sch$ corresponds to 
a very simple concept language, whose concepts are built from atomic
concepts and nominals using only intersection. 
In this sense, $L_\sch$ can be considered to be a fragment of  \textit{DL-Lite}$_{core,\sqcap}$ with nominals (also known as \textit{DL-Lite}$_{horn}$~\cite{arta:dllite09}), i.e., 
the description logic obtained by enriching \textit{DL-Lite}$_{core}$ (the simplest
language in the DL-Lite family) 
with conjunction.

The precise semantics of $L_\sch$ is as follows.
Given a concept $C$ that is defined in $L_\das$ and an 
instance $I$ over $\sch$, 
the extension of $C$ in $I$, denoted by $\ex{C}{I}$, is inductively defined below.
Intuitively, the extension of $C$ in $I$
is the result of evaluating the 
query associated with $C$ over $I$.
\[\small\begin{array}{@{}l@{\,}l@{\,}l@{}}
\ex{R}{I} &=& R^I \\ 
\ex{\sigma_{A_1 \mathtt{op}_1 c_1, \ldots, A_n \mathtt{op}_n c_n}(R)}{I} &=& \{\bar{b}\in R^I \mid \pi_{A_i}(\bar{b})\mathtt{op}_i c_i , 1\leq i\leq n \} \\
\ex{\top}{I} &=& \const \\
\ex{\nom{c}}{I} &=& \{c\} \\ 
\ex{\pi_A(D)}{I} &=& \pi_A(\ex{D}{I}) \\
\ex{C_1 \sqcap C_2}{I} &=& \ex{C_1}{I} \cap \ex{C_2}{I} 
\end{array}\]

\begin{figure*}
\begin{center}
\small
\begin{tabular}{lll}
{\bf $L_\das$ concept expression} & {\bf \textsc{select-from-where} formulation} & {\bf Intuitive meaning} \\[.5em]
$\pi_\text{name}(\text{Cities})$ & name from Cities & City        \\
$\pi_\text{name}(\sigma_{\text{continent} = \text{``Europe''}}(\text{Cities}))$ & name from Cities where continent=``Europe'' & European City  \\
$\pi_\text{name}(\sigma_{\text{continent} = \text{``N.America''}}(\text{Cities}))$ & name from Cities where continent=``N.America'' & N.American City \\
$\pi_\text{name}(\sigma_{\text{population}>1000000}(\text{Cities}))$ & name from Cities where population>1000000 & Large City\\
$\pi_1(\text{BigCity})$ & name from BigCity &  name of BigCity \\
$\{\text{``Santa Cruz''}\}$ & ``Santa Cruz'' & Santa Cruz  \\
\pbox{20cm}{$\pi_\text{name}(\sigma_{\text{population}<1000000}(\text{Cities}))\sqcap$ \\ $\pi_\text{city\_to}(\sigma_{\text{city\_from}=\text{Amsterdam}}(\text{Reachable}))$ }& \pbox{20cm}{name from Cities where population<1000000 \\ AND city\_from from Reachable where city\_to=Amsterdam} & Small City that is reachable from Amsterdam.
\end{tabular}
\end{center}
\vspace{-3mm}
\caption{Example of concepts specified in $L_\das$.}
\label{fig:concept-examples}
\end{figure*}

The notion of when one concept is subsumed by another
is defined according to the extensions of 
the concepts. There are two notions, corresponding to concept subsumption w.r.t. an instance
or subsumption w.r.t. a schema.
More precisely, given two concepts $C_1,C_2$, 
\begin{itemize}
\item we say that {\em $C_2$ subsumes $C_1$ w.r.t. 
an instance $I$} (notation: $C_1 \sqsubseteq_I C_2$) if $\ex{C_1}{I} \subseteq \ex{C_2}{I}$.

\item we say that {\em $C_2$ subsumes $C_1$ w.r.t. a 
schema} $\das$ (notation: $C_1 \sqsubseteq_{\das} C_2$), if for every instance $I$ of $\das$, we have that 
$C_1 \sqsubseteq_I C_2$.

\end{itemize}

We are now ready to define the two types of ontologies, which are based on the two notions of 
concept subsumption described above, that can be derived from an instance or a schema.

\begin{definition}[Ontologies derived from a schema]
Let $\das$ be a schema, and let $I$ be an instance of
$\das$.  Then the ontologies derived from $\das$ and $I$ are defined
respectively as

\begin{itemize}
\item $\ontocal_\das =  \sontology{L_\das}{\sqsubseteq_\das}{ext}$  and 

\item $\ontocal_I = \sontology{L_\das}{\sqsubseteq_I}{ext}$,  

\end{itemize}
where $ext$ is the function given by $ext(C,I') = \ex{C}{I'}$ 
for all instances $I'$ over $\das$. 
By $\ontocal_\das[\calK]$ we denote the ontology
$\sontology{L_\das[\calK]}{\sqsubseteq_\das}{ext}$,
and by $\ontocal_I[\calK]$ we denote the
ontology
$\sontology{L_\das[\calK]}{\sqsubseteq_I}{ext}$. 
\end{definition}

It is easy to verify 
that the subsumption relations $\sqsubseteq_{\das}$
and $\sqsubseteq_I$ are indeed pre-orders (i.e., reflexive, and
transitive relations), and that, for every fixed 
schemas $\das$, the function
$\ex{C}{I'}$ is polynomial-time computable. Hence, the above
definition is well-defined even though the ontologies obtained in this
way are typically infinite.
From the definition, it is easy to verify that if 
$C_1 \sqsubseteq_\das C_2$, then \mbox{$C_1\sqsubseteq_I C_2$}.

The following result about deciding $\sqsubseteq_I$ is immediate, as one can always execute the queries 
that are associated with the concepts
and then test for subsumption, which can be done in polynomial time.

\begin{restatable}{prop}{subsinstance}\label{subs-instance}
The problem of deciding, given an instance $I$ of a schema $\das$
and given two $L_\das$ concept expressions $C_1$, $C_2$, 
whether $C_1 \sqsubseteq_I C_2$, is in $\PTIME$.
\end{restatable}

On the other hand, the complexity of deciding $\sqsubseteq_{\das}$  depends on the type 
of integrity constraints that are used in the specification of $\das$.
Table~\ref{tbl:subsumption-complexity} provides a summary of relevant complexity results.

\newcommand{\csdecision}{\textsc{C-subsumption}}
\begin{table}
\small
\begin{center}
\begin{tabular}{|l|l|}
\multicolumn{1}{l}{Constraints} & 
\multicolumn{1}{l}{Complexity of subsumption for $L_\das$} \\
\hline
$\UCQ$-view def. (no comparisons) & $\NP$-complete \\
$\UCQ$-view def.   &   $\piptwo$-complete  \\ 
linearly nested $\UCQ$-view def. & $\piptwo$-complete \\
nested $\UCQ$-view def. & $\coNEXPTIME$-complete \\
FDs &  in $\PTIME$  \\
IDs & ? (in \PTIME for selection-free $L_\das$) \\
IDs + FDs & Undecidable \\
\hline
\end{tabular}
\end{center}
\normalsize
{All stated lower bounds already hold for $L^{\mathrm{min}}_\das$ concept expressions.
}
\vspace{-4mm}
\caption{Complexity of concept subsumption.}
\label{tbl:subsumption-complexity}
\vspace{-2mm}
\end{table}

\begin{restatable}{thm}{}
\label{th:table}
Let $\mathcal{W}$ be one of the different classes of schemas with integrity constraints
listed in Table~\ref{tbl:subsumption-complexity}. 
The complexity of the problem to decide, given a schema $\das$ 
in $\mathcal{W}$
and two  $L_\das$ concept expressions $C_1$, $C_2$, whether
 $C_1 \sqsubseteq_\das C_2$, is as indicated in 
 the second column of the corresponding row in Table~\ref{tbl:subsumption-complexity}. 
\end{restatable}

For example, given two concepts $C_1$, $C_2$, and
a schema $(\sch, \Sigma)$ where $\Sigma$ 
is a collection of nested $\UCQ$-view definitions, the complexity of deciding $C_1 \sqsubseteq_\das C_2$
is $\coNEXPTIME$-complete.  The lower bound already holds for concepts specified
in $L^{\mathrm{min}}_\das$.
We conclude this section with an analysis of the number of 
distinct concepts that can be formulated in a given concept language and an example that illustrates
explanations that can be computed from such derived ontologies.

\begin{restatable}{prop}{sizeconceptlanguage}\label{size-concept-language} Given a schema $\das$ and a finite set of constants $\mathcal{K} \subset \const$, the number of unique concepts (modulo logical equivalence) 
\begin{itemize}
\item  in $\LWmin[\calK]$ is polynomial in the size of $\das$ and $\calK$,
\item in selection-free or intersection-free $L_\das[\calK]$ is single exponential in the size of $\das$ and $\calK$.
\item in $L_\das[\calK]$ is double exponential in the size of $\das$ and $\calK$.
\end{itemize}
\end{restatable}

\begin{example}
\label{ex:explanationsWithDerived}
Let $\das$ and $I$ be the schema and instance from
Figure~\ref{fig:example-dataworkspace} and Figure~\ref{fig:train-instance}. Suppose the concept language $L_\das$
is used to define among others the concepts from
Figure~\ref{fig:concept-examples}. 
The following concept subsumptions can be derived from $\das$. 
Note that subsumption $\sqsubseteq_\das$ implies $\sqsubseteq_I$. 
\[
\small
\begin{array}{lll}
\pi_\text{name}(\sigma_{\text{continent} = \text{``Europe''}}(\text{Cities})) & \sqsubseteq_\das & \pi_\text{name}(\text{Cities}) \\ 
\pi_{\text{name}}(\sigma_{\text{population>7000000}}(\text{Cities})) & \sqsubseteq_\das & \pi_\text{name}(\text{BigCity}) \\
\pi_\text{name}(\text{BigCity}) & \sqsubseteq_\das & \pi_\text{name}(\text{Cities}) \\
\pi_\text{name}(\text{BigCity})  & \sqsubseteq_\das &  \pi_{\text{city\_from}}(\text{Train-Connections}) \\
\end{array}
\normalsize
\]
The first and second subsumptions follow from definitions.
The third one holds because according to $\Pi$, a BigCity is a city with population more than 5 million.
The fourth subsumption follows from the inclusion dependency that 
each BigCity must have a train departing from it. 
There are subsumptions that hold in 
$\ontocal_I$ but not in  $\ontocal_\das$. For instance,
\begin{tabbing}
$\pi_\text{city\_to}(\sigma_{\text{city\_from}=\text{Amsterdam}}(\text{Reachable})) \sqsubseteq_I$ \\
\hspace{4cm}$\pi_\text{city\_to}(\sigma_{\text{city\_from}=\text{Berlin}}(\text{Reachable})),$
\end{tabbing}
\noindent
holds w.r.t. $\ontocal_I$, where $I$ is the instance given in  Figure~\ref{fig:train-instance}, but does 
not hold w.r.t $\ontocal_\das$, since one can construct an instance
where not all cities that are reachable from Amsterdam are reachable from Berlin. 

We now give examples of most-general explanations w.r.t. $\ontocal_\das$ and $\ontocal_I$. 
As before,  let $q(x,y)$ = $\exists z.$ Train-Connections($x,z$) $\wedge$ Train-Connections($z,y$) be a query with
 $q(I)$ = $\{\tupp{\text{Amsterdam, Rome}},\tupp{\text{Amsterdam, Amsterdam}},$
 $ \tupp{\text{Berlin, Berlin}}, \tupp{\text{New York, Santa Cruz }}\}$. 
We would like to explain why 
$\langle$Amsterdam, New York$\rangle \not\in q(I)$ using the derived ontologies $\ontocal_\das$ and $\ontocal_I$. 
Note that if $E$ is an explanation w.r.t. $\ontocal_\das$, then it is also 
an explanation w.r.t. $\ontocal_I$ and vice versa. 
Some possible explanations are: 
\[\small\begin{array}{lll}
E_1 = \langle\pi_{\text{name}}(\sigma_{\text{continent=Europe}}(\text{Cities})), \\
\hspace{2cm}\pi_{\text{city\_from}}(\sigma_{\text{city\_to = San Francisco}}(\text{Train-Connections}))\rangle \\
E_2 = \langle\pi_{\text{name}}(\sigma_{\text{continent=Europe}}(\text{Cities})), \\
\hspace{2cm}\pi_{\text{name}}(\sigma_{\text{continent=N.America}}(\text{Cities}))\rangle \\
E_3 = \langle  \pi_{\text{city\_to}}(\sigma_{\text{city\_from = Berlin}}(\text{Reachable}))  ,  \\
\hspace{2cm} \pi_{\text{city\_from}}(\sigma_{\text{city\_to = Santa Cruz}}(\text{Reachable}))\rangle \\
E_4 = \langle\{\text{Amsterdam}\}, \pi_{\text{name}}(\sigma_{\text{population>7000000}}(\text{Cities})) \rangle \\
E_5 = \langle \pi_{\text{name}}(\sigma_{\text{country=Netherlands}}(\text{Cities})) , \\
\hspace{2cm} \pi_{\text{name}}(\text{BigCity})\sqcap  \pi_{\text{name}}(\sigma_{\text{continent=N.America}}(\text{Cities})) \rangle \\ 
E_6 = \langle \{\text{Amsterdam}\},  \{ \text{New York}\} \rangle \\

E_7= \langle\pi_{\text{name}}(\sigma_{\text{continent=Europe}}(\text{Cities})), \pi_{\text{name}}(\text{BigCity})\} \rangle \\
E_8= \langle\pi_{\text{name}}(\sigma_{\text{continent=Europe}}(\text{Cities})), \\
\hspace{2cm} \pi_{\text{name}}(\sigma_{\text{population>7000000}}(\text{Cities}))\} \rangle \\
\end{array}\]
For example, $E_1$ states the reason is that Amsterdam is a European city and New York is a city that has a train connection to San Francisco, and there is no train connection between such cities via a city.  The trivial explanation $E_6$ is less general than any other explanation 
w.r.t  $\ontocal_\das$ (and  $\ontocal_I$ too). 
It can be verified that
$E_2$ and $E_7$ are most-general explanations w.r.t both $\ontocal_\das$ and $\ontocal_I$.
In particular, $E_2 >_{\ontocal_I} E_5$ and $E_2\geq_{\ontocal_I} E_3$, but
$E_2\not >_{\ontocal_\das} E_5$ and $E_2\not >_{\ontocal_\das} E_3$
since there might be an instance of $\das$ where Netherlands is not in
Europe or where Berlin is reachable from a non-european city.  
\hfill$\Box$ 
\end{example}

In general, if $E$ is an explanation w.r.t. $\ontocal_I$ then
$E$ is also an explanation w.r.t. $\ontocal_\das$, and vice versa. The following proposition also describes the relationship between most-general explanations w.r.t $\ontocal_\das$ and $\ontocal_I$.

\begin{restatable}{prop}{mgerelationships}
Let $\das$ be a schema, and let $I$ be an instance of $\das$. 
\begin{itemize}
\item[(i)] Every explanation w.r.t.~$\ontocal_\das$ is an explanation w.r.t.~$\ontocal_I$ and vice versa.
\item[(ii)] A most-general explanation w.r.t~$\ontocal_\das$ is not necessarily a most-general explanation w.r.t.~$\ontocal_I$, and likewise vice versa.
\end{itemize}
\end{restatable}

\begin{proof}
The statement $(i)$ follows from Definition~\ref{dfn:explanation} and the definition of $ext$ for $\ontocal_\das$ and $\ontocal_I$. That is, $ext$ is the same on the input instance $I$ for both $\ontocal_\das$ and $\ontocal_I$, and the conditions of Definition~\ref{dfn:explanation} use only the value of $ext$ on $I$.
Going back to Example \ref{ex:explanationsWithDerived}, $E_1$ is a most-general explanation w.r.t.
$\ontocal_\das$, but it is not a most-general explanation w.r.t. $\ontocal_I$ (since
$E_3$ is a strictly more general explanation than $E_1$ w.r.t. 
$\ontocal_I$).  Thus, the first direction of $(ii)$ holds.
For the other direction of $(ii)$, consider $E_8$ which is a most-general explanation w.r.t. $\ontocal_I$. But it holds that $E_7>_{\ontocal_\das}E_8$ and $E_7$ is an explanation. 
Note that $E_7$ and $E_8$ are equivalent w.r.t. $\ontocal_I$.
\end{proof}

\section{Algorithms for Computing Most-General Explanations}
\label{sec:AlgorithmsForWhyNotExplanations}

Next, we formally introduce the ontology-based why-not
problem, which was
informally described in Section~\ref{sec:TheWhyNotProvenanceProblem},
and we define algorithms for computing most-general
explanations.
We start by defining the notion of a why-not instance (or why-not question).

\begin{definition}[Why-not instance]
Let $\das$ be a 
schema, $I$ an instance of $\das$, $q$ an $m$-ary query over $I$ and
$\tuple{a}= \tup$ a tuple of constants such that $\tuple{a} \notin
q(I)$. We call the quintuple $\wninstance$, where $Ans=q(I)$, a
\emph{why-not instance} or a {\em why-not question}. 
\end{definition}

In a why-not
instance, the answer set {\em Ans} of $q$ over $I$ is assumed to have been
computed already. This corresponds closely to the scenario under which 
why-not questions are
posed where
the
user requests explanations for why a certain tuple is missing in the output of
a query, which is computed a priori. Note that
since {\em Ans}=$q(I)$ is part of a why-not instance, the complexity of
evaluating $q$ over $I$ does not affect the complexity analysis of the
problems we study in this paper. In addition, observe that although a query $q$ is
part of a why-not instance, the query is not directly used 
in our derivation of explanations 
for why-not questions with ontologies. However, the general setup accomodates
the possibility to consider 
$q$ directly in the derivation of explanations and this is 
part of our future work.

We will study the following algorithmic problems concerning most-general explanations for
a why-not instance.

\begin{definition}
The \dwhynotp problem is the following decision problem: 
given a why-not instance $\wninstance$ and an $\sch$-ontology $\ontocal$ consistent with $I$, 
does there exist an explanation for  $\tuple{a}\not \in Ans$ w.r.t. $\ontocal$? 
\end{definition}

\begin{definition}
The \mgep problem is the following decision problem: 
given a why-not instance $\wninstance$, an $\sch$-ontology $\ontocal$ consistent with $I$, 
and a tuple of concepts $(C_1, \ldots, C_n)$, is the given tuple of
concepts a most-general explanation w.r.t. $\ontocal$ for $\tuple{a}\not\in Ans$? 
\end{definition}

\begin{definition}
The \whynotp problem is the following computational problem: 
given a why-not instance $\wninstance$ and an $\sch$-ontology $\ontocal$ consistent with $I$,
find a most-general explanation w.r.t. $\ontocal$ for
$\tuple{a}\not \in Ans$, if one exists.  
\end{definition}

Note that deciding the existence of an explanation w.r.t. a finite $\sch$-ontology is equivalent to
deciding existence of a most-general explanation w.r.t. the same $\sch$-ontology.

Thus,
our approach to the why-not problem makes use of $\sch$-ontologies. In
particular, our notion of a ``best explanation'' is a \emph{most-general
explanation}, which is defined with respect to an $\sch$-ontology. We
study the problem in three flavors: one in which the $\sch$-ontology
is obtained from an external source, and thus it is part of the input,
and two in which the $\sch$-ontology is not part of the input, and is
derived, respectively, from the schema $\das$, or from
the instance $I$.  

\subsection{External Ontology}
\label{sec:ExternalOntology}

We start by studying the case of computing ontology-based why-not
explanations w.r.t. an external $\sch$-ontology. 
We first study the
complexity of deciding whether or not there exists an explanation w.r.t. 
an external $\sch$-ontology.

\vfill\eject
\begin{restatable}{thm}{npcpdec}\label{np-comp-decision} ~
\begin{enumerate}
\item 
The problem \mgep is solvable in \PTIME. 
\item
The problem \dwhynotp is \NP-complete. It remains \NP-complete even for bounded schema arity.
\end{enumerate}
\end{restatable}

Intuitively, to check if a tuple of concepts is a most-general explanation, we can first check in \PTIME if it is an explanation. Then, for each concept in the explanation, we can check in \PTIME if it is subsumed by some other concept in $\ontocal$ such that by replacing it with this more general concept, the tuple of concepts remains 
an explanation.
The membership in \NP is due to the fact that we can guess a tuple
of concepts of polynomial size and verify in \PTIME that it is  an explanation. 
The lower bound is by a reduction from the
\textsc{Set Cover} problem.
Our reduction uses a query of unbounded arity and a schema of bounded
arity. As we will show in Theorem~\ref{correct-running-naive}, the 
problem is in \PTIME if the arity of the query is fixed.

In light of the above result, we define an algorithm, called the $\naivealgo$, which is an \EXPTIME
algorithm for solving the \whynotp problem.

\begin{algorithm}
\caption{\naivealgo}
   \KwIn{a why-not instance $\wninstance$, where $\tuple{a} = (a_1, \ldots, a_m)$, a finite $\sch$-ontology $\ontocal =\sontoshort$}
   \KwOut{the set of most-general explanations for $\tuple{a}\not\in Ans$ wrt $\ontocal$}
   
		Let $\mathcal{C}(a_i)=\{C\in \mathcal{C} \mid a_i\in ext(C,I) \}$ for all $i, 1\leq i\leq m$

		Let $\mathcal{X}=\{(C_1, \ldots , C_m) \mid C_i \in \mathcal{C}(a_i) \text{ and } (ext(C_1,I)\times \ldots\times ext(C_m,I))\cap Ans =\emptyset \}$

		\ForEach{ pair of explanations $E_1$,$E_2 \in \mathcal{X}$, $E_1\neq E_2$}
		{
		\If{ $E_1 >_\ontocal E_2$}
		{remove $E_2$ from $\mathcal{X}$}
		}
		return $\mathcal{X}$
\label{alg:computeall}
\end{algorithm}

This algorithm first generates the set of all
possible explanations, and then iteratively reduces the set by
removing
the tuples of concepts that are less general than some tuple of
concepts in the set. In the end, only most-general explanations
are returned.
At first, in
line 1, for each element of the tuple $\tuple{a} = \tup$, 
we build the set $\mathcal{C}(a_i)$ containing all the concepts in $\cal C$ whose
extension contains $a_i$.  Then, in line 2, we build the set of all
possible explanations by picking a concept in $\mathcal{C}(a_i)$ for
each position in $\tuple{a}$, and by discarding the ones that have a
non empty intersection with the answer set $Ans$. Finally, in lines
3-5,  we remove from the set those explanations that have a strictly
more general explanation in the set.

We now show that \naivealgo is correct (i.e. it outputs the set of all
most-general explanations for the given why-not instance w.r.t. to the given 
$\sch$-ontology), and
runs in exponential time in the size of the
input.

\begin{restatable}
{thm}{naive}
\label{correct-running-naive}
Let the why-not instance $\wninstance$ and the $\sch$-ontology $\ontocal$ be an input to \naivealgo and let $\cal X$ be the corresponding output. The following hold:
\begin{enumerate}
	\item $\mathcal{X}$ is the set of all most-general explanations for $\tuple{a}\not \in Ans$ (modulo equivalence);
	\item \naivealgo runs in \EXPTIME in the size of the input (in \PTIME if we fix the arity of the input query).
\end{enumerate}
\end{restatable}

Theorem~\ref{correct-running-naive}, together with Theorem~\ref{dl-lite-gav-complexity}, yields the following corollary
(recall that, by construction of $\ontocal_\obda$, it holds that every input instance $I$
is consistent with $\ontocal_\obda$). 

\begin{corollary}
There is an algorithm that takes as input a why-not instance $\wninstance$ and an OBDA specification $\obda=\bspec$, where $\mathcal{T}$ is a \dllite TBox and $\cal M$ is a set of GAV mappings, and computes all the most-general explanations for $\tuple{a}\notin Ans$ w.r.t. the $\sch$-ontology $\ontocal_\obda$ in \EXPTIME in the size of the input (in \PTIME if the arity of the $q$ is fixed) .
\end{corollary}

\subsection{Ontologies from an instance}
\label{sec:OntologyDerivedFromWorkspaceInstance}

We now study the why-not problem w.r.t. an $\sch$-ontology $O_I$ 
that is derived from an instance.
First, note that the presence of nominals in the concept language
guarantees a trivial answer for the \dwhynotproblem{$\ontocal_I$} problem. 
An explanation always exists, namely the explanation
with nominals corresponding to the constants of the tuple $\tuple{a}$. 
In fact, a \emph{most-general explanation} always exists, as  follows
from the results below. \looseness=-1

\begin{definition}
The 
\whynotproblem{$\ontocal_I$} is the following computational problem: 
given a why-not instance $\wninstance$, find 
a most-general explanation w.r.t. $\ontocal_I$ for $\tuple{a}\not \in Ans$, 
where $\ontocal_I$ is the $\sch$-ontology that is derived from $I$, as defined in Section~\ref{DerivingOntology}. 
\end{definition}

First, we state an important proposition, 
that underlies the correctness of the algorithms that we will present. The following
proposition shows that, when we search for explanations w.r.t. $\ontocal_I$, we can always restrict our attention to a particular finite restriction of this ontology.

\begin{restatable}{prop}{activedomain}\label{prop-active-domain}
Let $\wninstance$ be a why-not instance. 
If $E$ is an explanation for $\tuple{a}\not \in Ans$ w.r.t. $\ontocal_I$ (resp. $\ontocal_\das$),  then there exists an explanation $E'$ for $\tuple{a}\not \in Ans$ such that $E <_{\ontocal_I[\calK]} E'$ (resp. $E <_{\ontocal_\das[\calK]} E'$),
where $\mathcal{K} = \adom{I} \cup \{a_1, \ldots, a_m\}$ and each constant in $E'$ belongs to $\mathcal{K}$. 
\end{restatable}

In our proof, we iteratively reduce the number of constants occurring in
the explanation. That is, for every explanation $E$ with concepts
containing constants outside of $\adom{I} \cup \{a_1, \ldots, a_m\}$,
we produce a new explanation $E'$ which is more general than $E$ and
which contains less constants outside of $\adom{I} \cup
\{a_1, \ldots, a_m\}$. 

Notice that since, in principle, it is possible to materialize the
ontology $\ontocal_I[\calK]$ (i.e., to explicitly compute all the
concepts $\cal C$ in the ontology, the subsumption relation
$\sqsubseteq_I$, and the extension $ext$), the \naivealgo, together with
Proposition~\ref{prop-active-domain}, give us a method for solving
\whynotproblem{$\ontocal_I$}.
In particular, given a  schema, \naivealgo solves
\whynotproblem{$\ontocal_I$} in 2\EXPTIME (in \EXPTIME if the arity of
$q$ is fixed). This is because
to find a most-general explanation w.r.t $\ontocal_I$, it is sufficient to
restrict to  the concept language $L_\das[\calK]$ and its
fragments, where $\mathcal{K} = \adom{I}\cup\{a_1,\ldots,a_m\}$. Then
\whynotproblem{$\ontocal_I$} is solvable in 2\EXPTIME 
follows from the fact that the $\sch$-ontology $\ontocal_I[\calK]$ is
computable in at most 2\EXPTIME. 

We now present a more effective algorithm for solving
\whynotproblem{$\ontocal_I$}. (See
Algorithm~\ref{alg:greedy}.)
We start by introducing the notion of a \textit{least upper bound} of a set of constants $X$ w.r.t. an
instance $I$, denoted by
$\textsf{lub}_I(X)$. This, intuitively, corresponds to the most-specific concept
whose extension contains all constants of $X$.  
We first consider the case in which $\textsf{lub}_I(X)$ is expressed using selection-free $L_\das$ concepts.
The following lemma  
states two important properties of $\textsf{lub}_I(X)$ that are crucial for the correctness of Algorithm~\ref{alg:greedy}.

\begin{restatable}{lem}{lubnosel}\label{lub-no-selection}
Given an instance $I$ of schema $\sch$ and a set of constants $X$, we can compute in polynomial
time a selection-free $L_\das$ concept, denoted $\textsf{lub}_I(X)$, that is the smallest concept whose extension contains all the elements in $X$ definable in the language. In particular, the following hold:
\begin{enumerate}
   \item$X \subseteq ext(\textsf{lub}_I(X),I)$,
   \item there is no concept $C'$ in selection-free  $L_\das$ such that $C' \sqsubset_I \textsf{lub}_I(X)$ and  $X \subseteq ext(C',I)$.
\end{enumerate}
\end{restatable}
 
We are now ready to introduce the algorithm. We will start with a high-level description of the idea behind it. 
The algorithm navigates through the search space
of possible explanations using an incremental search strategy and
makes use of the above defined notion of $\textsf{lub}$. 
We start with an explanation that has, in each position, the $\textsf{lub}$ of the
constant (i.e., nominal) that occurs in that position.
Then, we try to construct a more general
explanation by expanding 
the set of constants considered by each $\textsf{lub}$.

Notice that \greedyalgo produces explanations which are tuples of conjunctions of concepts.
Therefore it produces an explanation whose concepts are concept
expressions in the language $L_\das$ or selection-free $L_\das$. We
will study the behavior of the algorithm in each of these cases
separately.

\begin{algorithm}
\label{alg:greedy}
\caption{\greedyalgo}
   \KwIn{a why-not instance $\wninstance$}
   \KwOut{a most-general explanation for $\tuple{a}\notin Ans$ wrt $\ontocal_I$} 

Let $\mathcal{K} = \adom{I}\cup\{a_1,\ldots,a_m\}$

Let $\mathcal{X} = (X_1, \ldots , X_m)$ s.t. each $X_j = \{a_j\}$. \textit{\scriptsize{// support set}}

Let $E= (C_1, \ldots , C_m)$ s.t. each $C_j = \textsf{lub}_I(X_j)$. \textit{\scriptsize{// first candidate explanation}}

\ForEach{$1\leq j\leq m$}
{

\ForEach{$b \in \adom{I} \setminus ext(E_j,I)$}
{

$X_j' = X_j \cup \{b\}$

Let $C'_j = \textsf{lub}_I(X_j')$ \textit{\scriptsize{// a more general concept in position $j$}}

Let $E' := (C_1, \ldots, C'_j, \ldots C_m)$ \textit{\scriptsize{// a more general explanation}}

\If{$E' \cap Ans = \emptyset$} 
{
$E := E'$

$\mathcal{X} := (X_1, \ldots, X'_j, \ldots X_m)$
}
}
}
return $E$
\end{algorithm}

First, we focus on the case in which \greedyalgo produces most-general explanations using selection-free $L_\das$ concepts. 
We show that the algorithm is correct, i.e., that it outputs an explanation for $\tuple{a}\notin Ans$ w.r.t. $\ontocal_I$, and that it runs in polynomial time with selection-free $L_\das$.

\begin{restatable}[Correctness and running time of \greedyalgo]{thm}{greedynosel}
\label{greedy-no-sel}
Let the why-not instance $\wninstance$ be an input to \greedyalgo and $E$ the corresponding output. 
The following holds:
\begin{enumerate}
	\item $E$ is a most-general explanation for $\tuple{a}\not \in Ans$ w.r.t. $\ontocal_I = (\mathcal{C},\sqsubseteq_I,ext)$, where $\cal C$ is selection-free $L_\das$;
	\item \greedyalgo runs in \PTIME in the size of the input.
\end{enumerate}
\end{restatable}

Now we extend our analysis of \greedyalgo to the general case in which it works with $L_\das$. 
First, we state an analogue of Lemma~\ref{lub-no-selection} for $L_\das$.

\begin{restatable}{lem}{lubsel}\label{lub-selection}
Given an instance $I$ of $\sch$ and a set of constants $X$, we can compute in exponential
time a $L_\das$ concept, denoted $\lubsigma{X}$, that is the smallest concept whose extension contains all the elements in $X$ definable in the language. Such concept is polynomial-time computable for bounded schema arity.
In particular, the following hold:
\begin{enumerate}
   \item$X \subseteq ext(\lubsigma{X},I)$,
   \item there is no concept $C'$ in $L_\das$ such that $C' \sqsubset_I \lubsigma{X}$ and  $X \subseteq ext(C',I)$.
\end{enumerate}
\end{restatable}

By \greedyalgosel we will refer to the algorithm obtained from \greedyalgo by replacing $\textsf{lub}_I(X)$ with $\lubsigma{X}$ in line 3 and line 7.

The following Theorem shows that \greedyalgosel is correct, i.e., that it outputs an explanation for $\tuple{a}\notin Ans$ w.r.t. the $\sch$-ontology $\ontocal_I$, and that it runs in exponential time (in polynomial time for bounded schema arity).

\begin{restatable}[Correctness and running time of \greedyalgosel]{thm}{greedysel}
\label{greedy-sel}
Let the why-not instance $\wninstance$ be an input to \greedyalgosel and $E$ the corresponding output. 
The following hold:
\begin{enumerate}
	\item $E$ is a most-general explanation for $\tuple{a}\not \in Ans$ w.r.t. $\ontocal_I = (\mathcal{C},\sqsubseteq_I,ext)$, where $\cal C$ is $L_\das$;
	\item \greedyalgo runs in \EXPTIME in the size of the input (in \PTIME for bounded schema arity).
\end{enumerate}
\end{restatable}

We close this section with the study of the following problem.

\begin{definition}
The \mgeproblem{$\ontocal_I$} problem is the following decision problem: 
given a why-not instance $\wninstance$ and a tuple of concepts $E = (C_1, \ldots, C_n)$, 
is $E$ a most-general explanation w.r.t. $\ontocal_I$ for $\tuple{a}\not\in Ans$? 
\end{definition}

Our next proposition states the running time of our algorithm for the \mgeproblem{$\ontocal_I$}
for various fragments of our concept language.
The algorithm 
operates very similarly to lines 4-11 of \greedyalgo. 
Given a tuple of concepts, we check whether that tuple of concepts can be extended to a more general tuple of 
concepts through ideas similar to lines 4-11 of \greedyalgo.
If the answer is ``no'', then we return ``yes''. Otherwise, we return ``no''. 

\begin{restatable}{prop}{checkmgeoi}\label{check-mge-o-i}
There is an algorithm that solves \mgeproblem{$\ontocal_I$} in: 
\begin{itemize}
	\item \PTIME for selection-free $L_\das$, or for $L_\das$ with bounded schema arity; 
	\item \EXPTIME for $L_\das$ in the general case.
\end{itemize}
\end{restatable}

\subsection{Ontologies from Schema}
\label{sec:OntologyDerivedFromTheWorkspaceSchema}

We now study the case of solving the why-not problem
w.r.t. to an $\sch$-ontology $O_\das$ that is derived
from a schema.  As in the previous case, 
the presence of nominals in the concept language guarantees that the
trivial explanation always exists. Therefore we do not 
consider
the decision problem \dwhynotproblem{$\ontocal_\das$}.

\begin{definition}[\whynotproblem{$\ontocal_\das$}]
The \whynotproblem{$\ontocal_\das$} is the following computational problem: 
given a why-not instance $\wninstance$, find a most-general explanation 
w.r.t. $\ontocal_\das$ for $\tuple{a}\not \in Ans$, where $\ontocal_\das$ is 
the $\sch$-ontology that is derived from $\das$, as defined in Section~\ref{DerivingOntology}. 
\end{definition}

The complexity of \whynotproblem{$\ontocal_\das$} depends on the complexity of subsumption checking for $L_\das$.
As seen in Table~\ref{tbl:subsumption-complexity}, subsumption checking with respect to arbitrary integrity constraints is undecidable. 
Therefore, for the general case in which no restriction is imposed on the integrity constraints, \whynotproblem{$\ontocal_\das$} is unlikely to be decidable.
The restrictions on the integrity constraints of $\das$ allow for the definition of several variants of the problem that, under some restrictions, are decidable. 

We restrict now to the cases in which we are able to materialize the $\sch$-ontology $\ontocal_\das[\calK]$, with $\calK =\adom{I}\cup\{a_1,\ldots,a_m\}$.
\naivealgo gives us a method for solving \whynotproblem{$\ontocal_\das$}. The following proposition  
gives us a double exponential upper bound for \whynotproblem{$\ontocal_\das$} in the general case, and a polynomial case under specific assumptions
(cf.~Table~\ref{tbl:subsumption-complexity}).

\begin{restatable}{prop}{whynotontow}
\label{prop:whynotOw}
There is an algorithm that solves \whynotproblem{$\ontocal_\das$}
\begin{itemize}
\item in 2\EXPTIME for
$L_\das$, provided that the input schema $\das$ is from a class for which  concept subsumption can be checked in $\EXPTIME$, 
\item in \EXPTIME  for selection-free $L_\das$, and projection-free $L_\das$, provided that the input schema $\das$ is from a class for which  concept subsumption  can be checked in $\EXPTIME$, 
\item  in \PTIME for $\LWmin$, if the arity of $q$ is fixed and provided that the input schema $\das$ is from a class for which concept subsumption  can be checked in $\PTIME$.
\end{itemize} 
\end{restatable}

We end with the definition of 
\mgeproblem{$\ontocal_\das$}.

\begin{definition}
The \mgeproblem{$\ontocal_\das$} problem is the following decision problem: 
given a why-not instance $\wninstance$ and a tuple of concepts $E=(C_1, \ldots, C_n)$, 
is $E$ a most-general explanation w.r.t. $\ontocal_\das$ for $\tuple{a}\not\in Ans$? 
\end{definition}

As for \whynotproblem{$\ontocal_\das$}, the undecidability of concept subsumption in the general case suggests that it is unlikely for \mgeproblem{$\ontocal_\das$} to be decidable without imposing any restriction on $\Pi$ and $\Sigma$. 
However, also this problem allows for the characterization of several decidable variants.

In particular, since \mgep is solvable in \PTIME (see Theorem~\ref{np-comp-decision}), by materializing $\ontocal_\das[\calK]$ we can derive some upper bounds for \mgeproblem{$\ontocal_\das$} too. 

\begin{restatable}{prop}{mgeontow}
There is an algorithm that solves \mgeproblem{$\ontocal_\das$}
\begin{itemize}
\item in 2\EXPTIME for $L_\das$ concepts, provided that the input  schema $\das$ is 
from a class for which concept subsumption can be checked in $\EXPTIME$, 
\item in \EXPTIME  for selection-free $L_\das$, and projection-free $L_\das$, 
provided that the input schema $\das$ is from a class for which concept subsumption can be checked in $\EXPTIME$, 
\item  in \PTIME for $\LWmin$, provided that the input  schema $\das$ is from a class for which concept subsumption can be checked in $\PTIME$.
\end{itemize} 
\end{restatable}

The proof is analogous to the one for Proposition~\ref{prop:whynotOw}.

We expect that the upper bounds for \whynotproblem{$\ontocal_\das$} and \mgeproblem{$\ontocal_\das$} can be improved. Pinpointing the complexity of these problems is left for future work.

\section{Variations of the Framework}
\label{sec:VariationsOfTheFramework}

We  consider several refinements and variations to our framework
involving 
finding short explanations, and providing alternative definitions of 
\emph{explanations} and of what it means to be \emph{most general}.

\smallskip
\par\noindent\textbf{Producing a Short Explanation.}
A most-general explanation that is \emph{short}
may be more helpful to the user.
To simplify our discussion, we restrict our attention to 
ontologies that are derived from an instance and show that the problem of finding
a most-general explanation of minimal length is $\NP$-hard in general,
where the {\em length} of an explanation $E=(C_1, \ldots, C_k)$ is measured by 
the total number of symbols needed to write out $C_1$, \ldots, $C_k$.
\label{sec:ProducingAMinimalExplanation}

\begin{restatable}{prop}{whynotminproblem}
\label{prop:whynotmin}
Given a why-not instance $\wninstance$, the problem of finding a most-general explanation to $\bar{a}\not \in Ans$ of minimal length is $\NP$-hard. 
\end{restatable}

Given that computing a shortest most-general explanation 
is intractable in general, we may consider
the task of shortening a given most-general explanation.
The \greedyalgo 
produces concepts that 
may contain superfluous conjuncts.  It is thus natural to ask whether 
the algorithm can be modified to produce a most-general explanation of
a shorter length.  This question can be formalized in at least two
ways.  

 Let $I$ be an
instance of a schema $\das$, and let 
$C = \sqcap\{C_1, \ldots, C_n\}$ be any $L_\das$ concept expression. We may assume
that each $C_i$ is intersection-free. We say that $C$ is
\emph{irredundant} if there is a no strict subset
 $X\subsetneq\{C_1, \ldots, C_n\}$ such that $C\equiv_{{\cal O}_I} \sqcap X$.
We say that an explanation (with respect to ${\cal O}_I$) 
is irredundant if it consists of irredundant concept expressions. 
We say that explanations $E_1$ and $E_2$ are \emph{equivalent} w.r.t. an ontology $\mathcal{O}$, denoted as $E_1\equiv_\mathcal{O}E_2$, if $E_1\leq_\mathcal{O} E_2$ and $E_2\leq_\mathcal{O} E_1$. 

\begin{restatable}{prop}{irredundunt}
\label{prop:irredundant}
There is a polynomial-time algorithm that takes as input an
instance $I$ of a schema $\das$, as well as 
an $L_\das$ concept expression $C$, and produces an 
irredundant concept expression $C'$ such that $C\equiv_{{\cal O}_I}C'$. 
\end{restatable}

Hence, by combining Proposition~\ref{prop:irredundant} with \greedyalgo, we can 
compute an irredundant most-general explanation w.r.t.~${\cal O}_I$ in 
polynomial time.

We say that an explanation $E=(C_1, \ldots, C_k)$ is \emph{minimized} w.r.t.
${\cal O}_I$ if there does not exist an explanation
$E'=(C_1, \ldots, C_k)$ such that $E \equiv_{{\cal O}_I} E'$ and  $E'$
is shorter than $E$. 
Every minimized explanation is irredundant, but the converse
may not be true.
For instance, let $O$ be an ontology with three 
atomic concepts $C_1, C_2, C_3$ such that $C_1\sqsubseteq_O C_2
\sqcap C_3$ and $C_2\sqcap C_3 \sqsubseteq_O C_1$.  Then the concept $C_2\sqcap
C_3 $ is irredundant with respect to $O$. However,
$C_1$ is an equivalent concept of strictly shorter length. 

\begin{restatable}{prop}{shortequiv}
Given a why-not instance $\wninstance$ and an explanation $E$ to why $\bar{a}\not \in Ans$, the problem of finding a minimized explanation equivalent to $E$ is $\NP$-hard.
\end{restatable}

\smallskip\par\noindent\textbf{Cardinality based preference.}
We have currently
defined a \emph{most-general explanation} to be an explanation $E$ such that
there is no explanation $E'$ with $E'>_{\cal O} E$. A natural alternative is to define
``most general'' in terms of 
the cardinality of the extensions of the concepts in an explanation.
Formally,
let ${\cal O}=\sontoshort$ be an $\sch$-ontology, and $I$ an instance. 
We define the \emph{degree of generality} of an explanation 
$E=(C_1, \ldots, C_m)$ with respect to ${\cal O}$ and $I$ to be
the (possibly infinite) sum $|ext(C_1,I)|+\cdots+|ext(C_m,I)|$. For two explanations,
$E_1, E_2$, we write
 $E_1>^{card}_{{\cal O},I} E_2$, 
 if $E_1$ has a strictly higher degree of generality than $E_2$ with respect to ${\cal O}$ and $I$.
 We say that an explanation $E$ is $>^{card}$-maximal (with respect to 
 ${\cal O}$ and $I$) if there is no 
 explanation $E'$ such that $E' >^{card}_{{\cal O},I} E$. 

\begin{restatable}{prop}{cardinality}
\label{prop:card}
Assuming P$\neq$NP, there is no \PTIME algorithm that takes as input a why-not instance $\wninstance$ and an $\sch$-ontology ${\cal O}$, and 
produces a  $>^{card}$-maximal explanation for $\tuple{a}\not\in Ans$. This holds even for unary queries.
\end{restatable} 

In particular, this shows  (assuming P$\neq$NP) that computing $>^{card}$-maximal
explanations is harder than computing most-general explanations.
The proof  of Proposition~\ref{prop:card} goes by reduction from a suitable variant of \textsc{Set Cover}.
Our reduction is in fact an $L$-reduction, which implies that
there is no $\PTIME$ constant-factor approximation algorithm for
the problem of finding a $>^{card}$-maximal explanation.

\smallskip\par\noindent\textbf{Strong explanations.}
We now examine an alternative notion of an explanation that is
essentially independent to the instance of a why-not question.  Recall
that the second condition of our current definition of an explanation 
$E=(C_1, \ldots, C_m)$
requires that $ext(C_1,I)\times\cdots\times ext(C_1,I)$ does not intersect with $Ans$, where $I$ is the given
 instance.  
We could replace this condition by a stronger condition, namely that
$ext(C_1, I')\times\cdots\times ext(C_1,I')$ does not intersect
with $q(I')$, for \emph{any} instance $I'$ of the given schema
that is consistent with the ontology ${\cal O}$.
If this holds, we say that $E$ is a \emph{strong explanation}.

A strong explanation
is also an explanation but not necessarily the other way round.
When a strong explanation $E$  for $\tuple{a}\not\in Ans$ exists, then, 
intuitively, the reason why $\tuple{a}$ does not belong to $Ans$, is
essentially 
independent from the specific instance $I$, and has
to do with the ontology ${\cal O}$ and the query $q$. In the case where the ontology 
${\cal O}$ is derived from a schema $\das$, 
a strong explanation may help one discover possible errors
in the integrity constraints of $\das$, or in the query $q$.
We leave the study of
strong why-not explanations for future work.\looseness=-1

\section{Conclusion}
\label{sec:Conclusion} 
We have presented a new framework for
why-not explanations, which leverages concepts from an ontology to
provide high-level and meaningful reasons for why a tuple is missing
from the result of a query.
Our focus in this paper was on developing a principled
  framework, and on identifying the key algorithmic problems.
The exact complexity of some problems raised in this paper remains open.
In addition, there are several direc\-tions for future work. 

Recall that, in general, there may be multiple most-general explanations for
$\tuple{a}\not\in q(I)$. 
While we have presented a polynomial time algorithm for computing a most-general explanation to a 
why-not question w.r.t. $\ontocal_I$ for the case of selection-free $L_\das$, 
the most-general explanation that is returned by the algorithm may not always be 
the most helpful explanation.
In future work, we plan to investigate whether
there is a polynomial delay algorithm for enumerating all most-general explanations for such ontologies.

Although we only looked at \emph{why-not explanations}, it will 
be natural to consider \emph{why explanations} in the context of an ontology,
and in particular, understand
whether the notion of most-general explanations, suitably adapted, applies in this setting.
In addition, Roy and Suciu~\cite{RoySuciu} recently initiated the study of 
what one could call ``why so high'' and ``why so low'' explanations for 
numerical queries (such as aggregate queries). Again, it would be interesting
to see if our approach can help in identifying high-level such explanations.

We have focused on providing why-not explanations to missing tuples of queries
that are posed against a database schema. 
However, our framework for answering the
why-not question is general and could, in principle, be applied also
to queries posed against the ontology in an OBDA
setting. 

Finally, we plan to explore ways whereby our high-level 
explanations 
can be used to complement and enhance
 existing data-centric and/or query-centric
approaches. We illustrate this with an example. 
Suppose a certain publication $X$ is missing from the 
answers to query over some publication database. 
A most-general explanation
may be that
X was published by Springer (supposing
all Springer publications are missing from the answers to the query). 
This explanation provides insight on potential 
high-level issues that may exist in the database and/or query.
For example, it may be that all
Springer publications are missing from the database (perhaps due to  
errors in the integration/curation process) or 
the query has inadvertently omitted the retrieval
of all Springer publications.
This is in contrast with existing data-centric (resp. query-centric) approaches,
which
only suggest  fixes
to the database instance (resp. query) so that the specific publication  $X$ appears
in the query result.

\newpage
\smallskip
\noindent
{\bf Acknowledgements~} We thank Vince B{\'a}r{\'a}ny, Bertram Lud\"ascher and Dan Olteanu for motivating
discussion during early stages of the research. 
Ten Cate is partially supported by NSF grant IIS-1217869.
Civili is partially supported by the EU under FP7 project Optique (grant n. FP7-318338).
Sherkhonov is supported by the Netherlands
 Organization for Scientific Research (NWO) under project number
 612.001.012 (DEX).
Tan is partially supported by NSF grant IIS-1450560. 

\balance
\bibliographystyle{abbrv}

\normalsize
\end{document}